\documentclass[a4paper,11pt]{article}

\usepackage[utf8]{inputenc}
\usepackage[T1]{fontenc}
\usepackage[english]{babel}

\usepackage{indentfirst}
\usepackage{amsmath}
\usepackage{amsthm}
\usepackage{amssymb}
\usepackage{graphicx}
\usepackage{subfigure}
\usepackage{psfrag}
\usepackage{color}
\usepackage{pgf,pgfnodes}

\allowhyphens

\newcommand{\valos}{\mathbb{R}}
\newcommand{\complex}{\mathbb{C}}

\newcommand{\eps}{\varepsilon}

\newcommand{\ordo}{\mathcal{O}}

\newtheorem{thm}{Theorem}

\newcommand{\bt}{\bar\theta}

\newcommand{\ket}[1]{{\left|#1\right\rangle}}
\newcommand{\bra}[1]{{\left\langle #1\right|}}

\newcommand{\vev}[1]{\left\langle #1 \right\rangle}

\usepackage{ifpdf}

\ifpdf
\usepackage{epstopdf}
\usepackage[pdftex,colorlinks,urlcolor=blue,citecolor=blue,linkcolor=blue]{hyperref}
\else
\usepackage[hypertex,colorlinks,urlcolor=blue,citecolor=blue,linkcolor=blue]{hyperref}
\fi

\makeatother

\begin{document}

\numberwithin{equation}{section}

\title{Form factor approach to diagonal finite volume matrix elements in Integrable QFT}
\author{Bal\'azs Pozsgay$^1$\\
~\\
 $^{1}$MTA-BME \textquotedbl{}Momentum\textquotedbl{} Statistical
Field Theory Research Group\\
1111 Budapest, Budafoki \'ut 8, Hungary
}

\maketitle

{\abstract{
We derive an exact formula for finite volume excited state mean values of local
operators in 1+1 dimensional Integrable QFT with diagonal scattering. Our result is a
non-trivial generalization of the LeClair-Mussardo series, which is a form factor
expansion for finite size ground state mean values. 
}}

\section{Introduction}

The study of finite volume effects is central to many areas of
theoretical physics including Quantum Field Theory (QFT) and Statistical
Physics. Understanding the volume dependence of physical observables
leads to efficient ways of extracting infinite
volume quantities, for example in 3+1 dimensional lattice QFT calculations. 
As a first step, one would like to understand how finite volume
effects influence the spectrum \cite{Luscher-Fmu1,Luscher-Fmu2}. As a
second step one can also consider composite objects like correlation
functions or matrix elements of local operators. An example is given
by transition matrix elements in lattice QFT, which were shown to
possess a non-trivial volume dependence \cite{Lellouch-Luscher}. 

Finite size effects are actively investigated in 1+1 dimensional
Integrable QFT, where the integrability of the models allows for an
exact determination of physical quantities. For example the exact
ground state energy (the Casimir-energy) is known in terms of
solutions of certain non-linear integral equations
\cite{zam-tba,klassen_melzer_tba1,DdV1,DdV2,blz-ddv}. Excited
state energies are also known in many cases
 \cite{DoreyTateoAnCont1,DoreyTateoAnCont2,ddv-exc,takacs-sine-Gordon-NLIE-exc,takacs-sineG-NLIE-odd,blz-exc}.
The techniques developed to study the finite size spectrum of 1+1 scattering theories
found applications and generalizations in the framework of the
AdS/CFT correspondence as well \cite{ads-cft-review}. 

In a Lorentz-invariant 1+1 dimensional field theory 
finite size effects are equivalent to finite temperature effects, as
can be seen by
choosing the compact dimension to lie in the  imaginary time
direction. 
Finite temperature correlation functions are relevant to real world
condensed matter experiments because 
the low-energy physics of certain materials leads to effective field theories which are
integrable in many cases \cite{essler-konik-review}. This motivated the study
of finite temperature
correlations in integrable QFT 
\cite{Leclair:1996bf,leclair_mussardo,Doyon:finiteTreview,fftcsa2,LM-sajat,Essler:2009zz}. 
In models with diagonal scattering the finite temperature one-point
functions are given  
 by the LeClair-Mussardo series
\cite{leclair_mussardo}. Even though this series has not yet been
proven from first principles, it is supported by strong
theoretical arguments \cite{Saleur:1999hq,fftcsa2,LM-sajat} and
numerical checks \cite{takacs-lm}. On the other hand, an
analogous result for the finite temperature two-point function is
still missing. The corresponding formula of \cite{leclair_mussardo}
is
ill-defined at higher orders and it was criticized in \cite{Saleur:1999hq,CastroAlvaredo:2002ud}.
The leading terms of a well-defined finite temperature expansion were derived in the
works \cite{Essler:2009zz,D22,Pozsgay:2009pv,takacs-szecsenyi-2p}, but the general pattern of
the higher correction terms is not clear yet.

One way of deriving finite temperature correlations is
through a finite volume regularization
\cite{fftcsa1,fftcsa2,LM-sajat,D22,Essler:2009zz}. In this approach it
is 
essential to know the volume dependence of the 
finite volume matrix elements of local operators. This problem was
solved in \cite{fftcsa1,fftcsa2} to all orders in $1/L$, where it was
shown that the finite volume form factors essentially coincide with
the corresponding infinite volume form factors, normalized by the
appropriate density of states in rapidity space. Exponential
corrections to the asymptotic results of \cite{fftcsa1,fftcsa2} were
considered in \cite{Pozsgay:2008bf}, where the so-called $\mu$-term of
the form factors 
was derived using finite volume bound state quantization (see also \cite{takacs-dangerous-mu}). 

In the present work we continue the line of research initiated in the
works \cite{fftcsa1,fftcsa2,Pozsgay:2008bf}. We consider a subset of
the finite volume matrix elements: excited state mean values. Even
though these objects are not of direct relevance to finite temperature
correlation functions, understanding the structure of the
higher exponential corrections in this special case might help to
derive results also for the off-diagonal matrix elements, which can
lead to a new way of obtaining finite temperature corrections to the
two-point function.

The structure of this paper is as follows. In Section 2 we collect the
basic results about finite volume QFT which will be used 
in the subsequent sections. In Section 3 we consider excited states in
finite volume and we formulate our basic conjecture for the excited
state mean values, which takes the form of a multiple integral series
where each integral runs over a non-trivial contour in the complex
plain. 
Sections 4 and 5 include calculations
needed to transform this result into a form where each
integral runs over the real axis only.
In Sections 6 and 7 we calculate our final formulas for the
one-particle and two-particle mean values. A general conjecture for
higher particle number is given in Section 8. Finally, Section 9
includes our conclusions. The reader who is not interested in the
 intermediate steps towards our main result in
Section 8 (equation \eqref{full-multiparticle-result}) may
skip sections 4,5,6 and 7. 

\section{Integrable QFT in finite volume -- basic ingredients}

Consider a massive Integrable Quantum Field theory in finite volume $L$
with periodic boundary conditions. The discrete spectrum of the Hamiltonian will
be denoted by $\ket{n}_L$ with the $n=0$ state being the vacuum. Energy
levels are denoted by $E_n(L)$. We assume that the energy density of
the vacuum is normalized to zero such that $E_0(L)$ is the
Casimir-energy which decays exponentially with the
volume. In the present work we limit
ourselves to theories with diagonal scattering. Moreover, for the sake
of simplicity we will only consider models with only one particle
species. 
The mass of the
single particle will be denoted by $m$.

Let $S(\theta)=e^{i\delta(\theta)}$ be the scattering matrix of the theory (a pure phase
in this case). It satisfies the relations
\begin{equation}
\label{S}
  S(\theta)S(-\theta)=1 \qquad\qquad S(\theta)=S(i\pi-\theta).
\end{equation}
For future use we define the derivative of the phase shift:
\begin{equation*}
  \varphi(\theta)=\frac{d}{d\theta} \left(-i \log S(\theta)\right).
\end{equation*}
It follows from \eqref{S} that
\begin{equation}
  \label{phieq}
\varphi(\theta)=\varphi(-\theta)=\varphi(i\pi+\theta).
\end{equation}

The finite volume ground state energy can be calculated by the
Thermodynamic Bethe Ansatz \cite{zam-tba}. It is given by
\begin{equation*}
  E_0(L)=-m\int \frac{d\theta}{2\pi} \cosh(\theta) \log(1+e^{-\eps_0(\theta)}),
\end{equation*}
where $\eps_0(\theta)$, the so-called pseudoenergy function is given
by the solution of the non-linear integral equation
\begin{equation}
\label{TBA}
   \eps_0(\theta)=e(\theta)L-
  \int_{-\infty}^\infty\frac{d\theta'}{2\pi} 
\varphi(\theta-\theta') \log(1+e^{-\eps_0(\theta')}),
\end{equation}
with $e(\theta)=m\cosh\theta$ being the one-particle energy and $L$ is
the volume.

In this work we are concerned with the finite volume mean values
\begin{equation}
\label{ezt}
 _L\bra{n}\ordo(x)\ket{n}_L, 
\end{equation}
where $\ket{n}_L$ is an exact eigenstate of the finite volume Hamiltonian
normalized to unity and
$\ordo(x)$ is a local operator of the theory, defined with the same
normalization as in infinite volume. Translation invariance implies
that the mean value does not depend on $x$, therefore this variable
will be omitted in the following. For the simplicity we only consider
scalar operators.

Our main goal is to derive a form factor expansion for the objects \eqref{ezt}.
 The (infinite volume) form factors are defined as the
matrix elements of the local operator on infinite volume multiparticle
states:
\begin{equation*}
  F_{nm}^\ordo(\theta_1',\dots,\theta_n'|\theta_1,\dots,\theta_m)=
\bra{\theta_1',\dots,\theta_n'}\ordo\ket{\theta_1,\dots,\theta_m}.
\end{equation*}
The multiparticle states are equal to the \textit{in} or \textit{out}
scattering states for a given ordering of the rapidities:
\begin{equation*}
  \ket{\theta_1,\dots,\theta_m}=
\begin{cases}
|\theta_{1},\dots,\theta_{m}\rangle^{in} & :\;\theta_{1}>\theta_{2}>\dots>\theta_{m}\\
|\theta_{1},\dots,\theta_{m}\rangle^{out} & :\;\theta_{1}<\theta_{2}<\dots<\theta_{m}.\end{cases}
\end{equation*}
All form factors can be expressed with the elementary form factors
using the crossing relation:
\begin{equation}
\label{crossing}
  F_{nm}^\ordo(\theta_1',\dots,\theta_n'|\theta_1,\dots,\theta_m)=
  F_{0,n+m}^\ordo(\theta_1'+i\pi,\dots,\theta_n'+i\pi,\theta_1,\dots,\theta_m).\qquad
\end{equation}
The above equation is valid whenever there are no coinciding
rapidities, otherwise there are also contact terms present. 

In many cases the form factors have been constructed explicitly using
the so-called form factor bootstrap program. 
The idea of this program is to construct all
form factor functions which satisfy a certain set of equations (also
called the ``form factor axioms'') and possess certain analyticity properties,
and to identify the solutions describing the actual form factors of a
given operator. Here we do not review this procedure but instead refer
the reader to \cite{smirnov_ff,Babujian:2006km}. We assume that the form factors of the local
operator in question are known (or can be calculated in
principle). In the calculations presented below 
we use only two of the general properties, namely
 the exchange axiom and the kinematical pole property
satisfied by the elementary form factors:
\begin{equation}
\label{exchange}
  F_n^\ordo(\theta_1,\dots,\theta_j,\theta_{j+1},\dots,\theta_n)=
S(\theta_j-\theta_{j+1})  F_n^\ordo(\theta_1,\dots,\theta_{j+1},\theta_{j},\dots,\theta_n)
\end{equation}
\begin{equation}
\label{kinpole}
-i\mathop{\textrm{Res}}_{\theta=\theta^{'}}F_{n+2}^{\mathcal{O}}(\theta+i\pi,\theta^{'},\theta_{1},\dots,\theta_{n})
=\left(1-\prod_{k=1}^{n}S(\theta'-\theta_{k})\right)F_{n}^{\mathcal{O}}(\theta_{1},\dots,\theta_{n}).
\end{equation}
These equations provide the basis for analyzing the diagonal limit of
the form factors needed to calculate the finite volume mean values.

\subsection{Previous results for the mean values}

In the case of the ground state mean value it is expected that the large volume
limit reproduces the infinite volume expectation value:
\begin{equation*}
  \vev{\ordo}_{L}  =   \vev{\ordo}+\dots,\qquad mL\gg 1.
\end{equation*}
The dots denote finite volume corrections which are of order
$e^{-mL}$. An all-order result 
was found by LeClair and Mussardo in 
\cite{leclair_mussardo} where they considered the equivalent problem
of finite temperature one-point functions. They found an infinite
integral series
 (also-called the LeClair-Mussardo series):
\begin{equation}
\label{LM}
    \vev{\ordo}_{L}=\sum_n \frac{1}{n!}
\int \frac{d\theta_1}{2\pi}\dots  \frac{d\theta_n}{2\pi}
\left(\prod_j \frac{1}{1+e^{\eps_0(\theta_j)}}\right)
F^\ordo_{2n,c}(\theta_1,\dots,\theta_n).
\end{equation}
The form factors entering the integrals are the connected parts of
the diagonal infinite volume form factors:
\begin{equation*}
    F^\mathcal{O}_{2n,c}(\theta_1\dots\theta_n)\equiv
\text{finite piece of }\
F_{2n}^\mathcal{O}(\theta_1+i\pi+\eta_1,\dots,\theta_n+i\pi+\eta_n,\theta_n,\dots,\theta_1).
\end{equation*}
In general the object on the r.h.s. includes terms proportional to
$\eta_j/\eta_k$ and the rule for the connected part is to discard
all these terms.
For a detailed discussion of this diagonal limit we refer the reader to
\cite{fftcsa2}. Important properties of the functions $F^\ordo_{2n,c}$
will be discussed in section \ref{sec:FF}.

The physical interpretation of \eqref{LM} is
the following: in the finite volume situation the local operator
can interact with an arbitrary number of virtual particles which wind
around the finite volume. The amplitude associated to these
processes is the properly defined limit of the infinite volume
form factor. The fact that the normalization factor associated to
these processes is just the product of the weight functions
$1/(1+e^{\eps_0(\theta)})$ is a highly nontrivial consequence of the
integrability of the theory. 
The volume $L$ only enters \eqref{LM} through the 
the pseudoenergy $\eps_0(\theta)$ which is the unique
solution of \eqref{TBA}.
A graphical interpretation of the
integral series is shown in Figure \ref{fig:LM}.

\begin{figure}
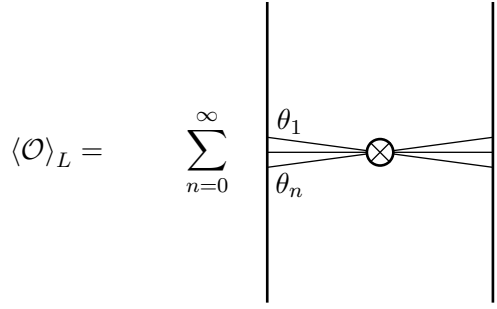

  \centering
\begin{pgfpicture}{6cm}{-1cm}{11cm}{4cm}

\pgfsetlinewidth{1pt}
\pgfline{\pgfxy(8,0)}{\pgfxy(8,4)}
\pgfline{\pgfxy(11,0)}{\pgfxy(11,4)}
\pgfcircle[stroke]{\pgfpoint{9.5cm}{2cm}}{5pt}
\pgfsetlinewidth{0.5pt}
\pgfline{\pgfxy(9.365,1.865)}{\pgfxy(9.635,2.135)}
\pgfline{\pgfxy(9.635,1.865)}{\pgfxy(9.365,2.135)}

\pgfline{\pgfxy(8,2.2)}{\pgfxy(9.34,2.02)}
\pgfline{\pgfxy(8,2)}{\pgfxy(9.34,2)}
\pgfline{\pgfxy(8,1.8)}{\pgfxy(9.34,1.98)}

\pgfline{\pgfxy(11,2.2)}{\pgfxy(9.66,2.02)}
\pgfline{\pgfxy(11,2)}{\pgfxy(9.66,2)}
\pgfline{\pgfxy(11,1.8)}{\pgfxy(9.66,1.98)}

\pgfputat{\pgfxy(8.3,2.25)}{\pgfbox[center,bottom]{$\theta_1$}}
\pgfputat{\pgfxy(8.3,1.4)}{\pgfbox[center,bottom]{$\theta_n$}}

\pgfputat{\pgfxy(7.2,2)}{\pgfbox[center,center]{$\displaystyle\sum_{n=0}^\infty$}}
\pgfputat{\pgfxy(5.2,2)}{\pgfbox[center,center]{$\vev{\ordo}_L=$}}

\end{pgfpicture}
\caption{Graphical interpretation of the LeClair-Mussardo formula for
  the finite volume ground state expectation values. Time runs in the vertical direction. In the
  horizontal direction periodic boundary conditions are understood. }
\label{fig:LM}
\end{figure}

We note that although \eqref{LM} is generally believed
to be true, a rigorous proof from first
principles is not yet available. In \cite{LM-sajat} it was proven to
all orders using
an expansion for finite volume form factors found in \cite{fftcsa2}
and to be presented below
(eq. \eqref{fftcsa2-result2}). However, the relation \eqref{fftcsa2-result2}
itself has not yet been proven (see the discussion below).

\bigskip

Mean values in finite volume excited states have been considered
previously in \cite{fftcsa2}. This work considered the IR limit, when
the states can be described with good approximation as Bethe Ansatz
states.
A first guess for the mean values in this approximation could
be that the mean value in a 
state $\ket{\theta_1,\dots,\theta_K}$ is simply the diagonal form
factor $F_{2K,c}^\ordo(\theta_1,\dots,\theta_K)$, possibly with a
normalization factor depending on $L$. However, the situation is more
complicated. In \cite{fftcsa2} it was found that there is expansion
for the mean value where each term  corresponds to a 
bipartite partitioning of the rapidities, where one subset of the
particles interacts with the operator and the remaining particles
only influence the normalization factor associated to this process. 

To
be specific, 
consider the  finite volume situation in the limit $mL\gg 1$ and a multiparticle
state described by Bethe roots $\{\theta_1,\dots,\theta_K\}$.
They
satisfy the Bethe equations
\begin{equation}
\label{logBY}
  Q_j=p_jL+\sum_{k\ne j}\delta(\theta_j-\theta_k)=2\pi
  I_j\qquad j=1\dots K,\quad I_j\in \mathbb{Z}.
\end{equation}
And important quantity is the density of states in rapidity space,
which is given by the Jacobian of the mapping given by \eqref{logBY}:
\begin{equation}
  \rho_K(\theta_1,\dots,\theta_K)=\det \mathcal{J}^{jk},
\qquad\qquad
\mathcal{J}^{jk}=\frac{\partial Q_j}{\partial \theta_k}.
\label{Jdef}
\end{equation}
In many non-relativistic models $\rho_K$ also describes the norm of
the Bethe Ansatz state and it is called the Gaudin determinant \cite{Gaudin-LL-norms,korepinBook}. 
 
We will also need the minors of the matrix $\mathcal{J}$.
For a
given bipartite partition 
\begin{equation*}
  \{\theta_1,\dots,\theta_K\} = \{\theta_+\}\cup \{\theta_-\}
\end{equation*}
\begin{equation*}
  \big|\{\theta_+\}\big|=K-n \quad\text{and}\quad \big|\{\theta_-\}\big|=n
\end{equation*}
we define the restricted determinant
\begin{equation}
\label{restricted-density}
  \tilde\rho_{K-n}(\{\theta_+\}|\{\theta_-\})=\det \mathcal{J}_+,
\end{equation}
where $\mathcal{J}_+$ is the sub-matrix of  $\mathcal{J}$ corresponding
to the particles in the set $\{\theta_+\}$.
Note that $\bar\rho_{N-n}(\{\theta_+\}|\{\theta_-\})$
still contains information about the complementary set of rapidities
$\{\theta_-\}$.

With these notations, the expression for the expectation
value reads
\begin{equation}
    \label{fftcsa2-result2}
\begin{split}
&\bra{\theta_1,\dots,\theta_K}\ordo\ket{\theta_1,\dots,\theta_K}_L=\\
&\hspace{2cm}\frac{1}{\rho_K(\theta_1,\dots,\theta_K)}
\sum_{\{\theta_+\}\cup \{\theta_-\}}
F^\ordo_{2n,c}\big(\{\theta_-\}\big)
\tilde\rho_{K-n}\big(\{\theta_+\}|\{\theta_-\}\big)+\ordo(e^{-\mu L}).
\end{split}
\end{equation}
Equation \eqref{fftcsa2-result2} is expected to describe the finite
size effects to all orders in $1/L$, beyond which there are only exponentially small 
corrections. The exponent $\mu$ is a mass scale proportional to $m$ which is
determined by the fusion processes in the theory \cite{klassen_melzer}. As remarked
earlier \eqref{fftcsa2-result2} has not yet been proven, but 
overwhelming numerical evidence has been gathered which support its
validity \cite{fftcsa2,takacs-palmai,takacs-dangerous-mu}. The
only analytical proofs available concern the cases $K=1$ 
and $K=2$ \cite{fftcsa2}.

\section{Excited states}

In many cases there are exact results available about the finite
volume energies of the excited states. 
The first paper to derive TBA equations for particle-like excited states in integrable QFT was
\cite{DoreyTateoAnCont1} where the authors used analytical
continuation in the volume parameter $L$ to cross from the ground
state to excited states with zero total momentum. This procedure was
motivated partly by closely related ideas in Quantum Mechanics \cite{Bender-Wu}.
Other works which derived excited state TBA equations include
\cite{blz-exc,DoreyTateoAnCont2,Teschner:2007ng}. 
In the known cases 
the only difference
between the excited state and ground state TBA is the addition of certain source
terms, or equivalently a modification of certain integral
contours. This property holds also for the various NLIE results, both
in field theory and lattice models \cite{ddv-exc,takacs-sine-Gordon-NLIE-exc,takacs-sineG-NLIE-odd,klumper-pearce1,klumper-pearce2}. 

Our approach to obtain excited state mean values  follows the
general idea of the papers \cite{DoreyTateoAnCont1,DoreyTateoAnCont2},
namely that there exists an analytic continuation procedure
which connects the excited states (or at least a subset of them) to the
ground state. Starting from the exact result \eqref{LM} for the ground
state expectation value it is natural to expect that a proper analytic
continuation of the full integral series yields the excited state mean
values. 

In \cite{DoreyTateoAnCont1} the analytic continuation was performed in
the volume parameter $L$.  The volume dependence of the series \eqref{LM} is only through
the pseudoenergy function $\eps_0(\theta)$ which is the solution of the TBA equation
\eqref{TBA}. 
In the course of the analytic continuation certain singularities of
the pseudoenergy function cross the real line, and this results 
in a necessary change of integration contour for the TBA equation.
We expect that the
same holds also for the mean values. Namely, the excited state
mean values will be given by equation  \eqref{LM} such that $\eps_0(\theta)$
has to be replaced by the solution of the appropriate
excited state TBA, and the integration contours for the multiple
integrals have to be changed accordingly.

In a generic theory a one-particle state is often represented in the excited
state TBA as a pair of complex rapidities which are complex
conjugates of each other. This is closely related to the 
fusion processes in the infinite volume theory. Accordingly it was
observed in \cite{Pozsgay:2008bf} that the leading exponential
corrections follow simply from the asymptotic formula \eqref{fftcsa2-result2} when
single particles are represented by an appropriate finite volume bound
state. This pattern of ``root doubling'' is expected to hold even in the exact result,
which would make
the formulas quite involved, as already the
one-root and two-root problems are technically complicated, as shown below.
Therefore in the remainer of this work
 we only consider the sinh-Gordon model, in which
excited state TBA equations of $K$ particles involve exactly $K$
complex roots  \cite{Teschner:2007ng} instead of $2K$ roots in a
typical case in models with particle fusion.
However, the analytic continuation procedure we intend to use
 is not established in the sinh-Gordon model; the paper
\cite{Teschner:2007ng} employed completely different
methods to derive the excited state TBA. Therefore the remainder of
this paper should be considered as a ``technical 
demonstration'' of how the calculations proceed: we just assume that there
is a proper analytic continuation and  we derive the excited state mean
values accordingly.
 We believe that our
results are valid in the sinh-Gordon model and that 
analogous formulas apply in those cases where single particles are
represented by a pair of complex rapidities. We give more remarks on
this in the Conclusions. 

Turning to the sinh-Gordon model we recall the results of
\cite{Teschner:2007ng}.
 Excited states can be characterized by a set of integer
quantum numbers $\{I_1,I_2,\dots,I_K\}$, $K\ge 1$ and a set of real rapidities
(Bethe roots) 
$\{\bar\theta_1,\bar\theta_2,\dots,\bar\theta_K\}$\footnote{We use the
  notation $\bar\theta_j$ for the rapidities 
  entering the excited state TBA equations in order to distinguish
  them from the auxiliary variables $\theta_j$ entering the multiple
  integrals for the mean values.}. The excited state TBA equation
reads
\begin{equation}
  \label{eq:excited_eps_sinhG}
   \eps(\theta)=mL\cosh\theta+
\sum_i \log S(\theta-\bar\theta_i-i\pi/2)
-\int  \frac{d\theta'}{2\pi}\varphi(\theta-\theta') \log(1+e^{-\eps(\theta')}).
\end{equation}
The condition for the Bethe roots is 
\begin{equation}
\label{barthetacond}
  \eps(\bar\theta_j+i\pi/2)=i(2I_j+1)\pi,\qquad
j=1\dots K, \qquad I_j\in\mathbb{Z}.
\end{equation}
The finite size energies are then given by 
\begin{equation}
  \label{eq:excitedE_sinhG}
  E=\sum_{j=1}^K m \cosh\bar\theta_j-\int \frac{d\theta}{2\pi}
  m\cosh(\theta) \log(1+e^{-\eps(\theta)}). 
\end{equation}

The condition \eqref{barthetacond} can be written explicitly as
\begin{equation}
  \label{eq:excited_eps_BY}
   mL \sinh\bar \theta_j+\sum_{k\ne j} \delta(\bar\theta_j-\bar\theta_k)+
\int
  \frac{d\theta'}{2\pi}i\varphi(\bar \theta_j-\theta'+i\pi/2)
  \log(1+e^{-\eps(\theta')})=
2I_j\pi.
\end{equation}
Equations \eqref{eq:excited_eps_BY} can be interpreted as Bethe equations
for the rapidities $\bar\theta_j$ modified by vacuum polarization
effects.

Given a state $\ket{\bar\theta_1,\dots,\bar\theta_K}$ we define an
integration contour $\mathcal{C}$ which consists of the real 
line and small circles clockwise around the points
$\bar\theta_j+i\pi/2$. Then the excited state TBA equations can be
re-written 
as
\begin{equation}
  \label{eq:excited_eps_sinhG2}
   \eps(\theta)=mL\cosh\theta
-\int_{\mathcal{C}}  \frac{d\theta'}{2\pi}\varphi(\theta-\theta')
\log(1+e^{-\eps(\theta')})
\end{equation}
\begin{equation}
  \label{eq:excitedE_sinhG2}
  E=-\int_{\mathcal{C}} \frac{d\theta}{2\pi} m\cosh(\theta) \log(1+e^{-\eps(\theta)}),
\end{equation}
where we used the fact that the function $(1+e^{-\eps(\theta)})$ has
simple zeroes at $\bar\theta_j+i\pi/2$, $j=1\dots K$.

As explained above, we conjecture that the mean values of local
operators in the given state can be expressed as
\begin{equation}
\label{LMexc}
\begin{split}
 &  \bra{\bar\theta_1,\bar\theta_2,\dots,\bar\theta_K} \ordo
  \ket{\bar\theta_1,\bar\theta_2,\dots,\bar\theta_K}_{L}=\\
&\hspace{3cm}
\sum_n \frac{1}{n!}
\int_{\mathcal{C}} \frac{d\theta_1}{2\pi}\dots  \frac{d\theta_n}{2\pi}
\left(\prod_j \frac{1}{1+e^{\eps(\theta_j)}}\right)
F^\ordo_{2n,c}(\theta_1,\dots,\theta_n),
\end{split}
\end{equation}
where the pseudoenergy function $\eps(\theta)$ and the roots
$\bar\theta_j$ are given as a
solution of the equations
\eqref{eq:excited_eps_sinhG}-\eqref{barthetacond}, and 
the contour $\mathcal{C}$ is determined by the roots $\bar\theta_j$.

Although the equation above is well-defined and expected to be exact,
it is not very enlightening. The remainder of this work is devoted to
the evaluation of the residues at the points
$\bar\theta_j+i\pi/2$ such that in our final formulas the integrals
only run over the real line. This leads to a representation where the
physical meaning of the individual terms is more transparent.

Note that evaluating all the residues 
 the rapidities $\bar\theta_j+i\pi/2$ will appear as
multiple insertions in the connected form factors. This motivates the
investigation of the degenerate cases for the form factors (section
\ref{sec:FF}).
The reader who is not interested in the
technical details of these calculations may skip the following four
sections and turn to section \ref{sec:final} which presents our final result
for arbitrary multiparticle states.

\section{Evaluating the residues}

Here we evaluate the contour integrals in formula \eqref{LMexc} and
express the result as a sum of integrals over the real line. For simplicity we only consider
one-particle and two-particle states; it is straightforward to
generalize these formulas to  higher particle number.

In writing down the multiple integrals we will frequently make use of
the shorthand
\begin{equation*}
\int \widetilde{d\theta} \equiv  \int \frac{d\theta}{2\pi} \frac{1}{1+e^{\eps(\theta)}}.
\end{equation*}
In all formulas below it is understood that $\eps$ is the solution of
the TBA equation corresponding to the finite volume state in
question. 

\subsection{One-particle states}

First we consider one-particle states, ie. $K=1$. In this case the
excited state TBA is simply
\begin{equation}
  \label{eq:excited_eps_sinhG1p}
   \eps(\theta)=mL\cosh\theta+
 \log S(\theta-\bar{\theta}-i\pi/2)
-\int  \frac{d\theta'}{2\pi}\varphi(\theta-\theta')
\log(1+e^{-\eps(\theta')}). 
\end{equation}
We define the function
\begin{equation}
\label{norm_fact}
\bar Q(\bar{\theta},L)=-i \eps(\bar{\theta}+i\pi/2)=mL\sinh\bar{\theta}-
\int
  \frac{d\theta'}{2\pi}i\varphi(\theta'-\bar{\theta}+i\pi/2)
  \log(1+e^{-\eps(\theta')}).
\end{equation}
Then the quantization condition for $\bar{\theta}$ 
is simply $\bar Q(\bar{\theta})=2\pi (I+1/2)$.

The integration contour $\mathcal{C}$ in \eqref{LMexc} consists of the
real line and a small circle around the point $\bar\theta+i\pi/2$. The
residue of the weight function at this point is
\begin{equation*}
 \text{Res}_{\theta=\bar\theta+i\pi/2} \frac{1}{1+e^{\eps(\theta)}}=
\left(-\left.\frac{\partial\eps(\theta)}{\partial\theta}
\right|_{\theta=\bar\theta+i\pi/2}\right)^{-1}.
\end{equation*}

Evaluating the residues in the multiple integral series 
we write
\begin{equation}
\label{LM1pkifejt}
  \bra{\bar\theta}\ordo \ket{\bar\theta}=\sum_{j,k=0}^\infty
\mathcal{L}_{jk}.
\end{equation}
Here $\mathcal{L}_{jk}$ represents the contributions where in the $n=j+k$ term
in the series we integrate $j$ times around $\bar\theta+i\pi/2$ and
$k$ times over the real line. Picking up the residues at
$\bar\theta+i\pi/2$ we obtain
\begin{equation}
\label{LM1pkifejt2}
  \mathcal{L}_{jk}=\frac{1}{j!k!} \int \widetilde{d\theta_1}\dots \widetilde{d\theta_k}
\frac{F^\ordo_{2(j+k),c}(\bt+i\pi/2,\bt+i\pi/2,\dots,\theta_1,\dots,\theta_k)}{n_1^j},
\end{equation}
where 
\begin{equation}
\label{n1def1}
n_1=(-i)\left.  \frac{\partial\epsilon(\theta)}{\partial \theta}\right|_{\theta=\bar\theta+i\pi/2}=
mL\cosh\bar\theta
+\varphi(0)+\int  \frac{d\theta'}{2\pi}i\varphi(\theta'-\bar{\theta}+i\pi/2)
\frac{1}{1+e^{\eps(\theta')}}
\left(\frac{\partial\eps}{\partial\theta }\right),
\end{equation}
and in the form factor above there are $j$ insertions of $\bar\theta+i\pi/2$.

It is useful to derive an integral series for the expressions
\eqref{n1def1}.
Introducing the kernel
\begin{equation}
\label{Kkernel}
  \hat{K}(\theta,\theta')=\varphi(\theta-\theta')\frac{1}{1+e^{\eps(\theta')}}
\end{equation}
and differentiating \eqref{eq:excited_eps_BY} one obtains
\begin{align*}
 (1-\hat{K})\frac{\partial\eps}{\partial\theta }=mL\sinh\theta-i\varphi(\theta-\bar\theta+i\pi/2).
\end{align*}
Introducing the resolvent $\hat{M}$ satisfying 
\begin{equation*}
  (1+\hat{M})(1-\hat{K})=1
\end{equation*}
we have
\begin{equation*}
\frac{\partial\eps}{\partial\theta }
=(1+\hat{M})\Big(mL\sinh\theta-i\varphi(\theta-\bar\theta+i\pi/2)\Big)\quad\quad
\frac{\partial\eps}{\partial\bar\theta }=(1+\hat{M})i\varphi(\theta-\bar\theta+i\pi/2).
\end{equation*}
Using $1+\hat{M}=\sum_{n=0}^\infty K^n$ we obtain the integral series
\begin{equation}
\begin{split}
\label{n1_szepen_kifejezve}
& n_1 =\varphi(0)+
mL\cosh\bar\theta
+\int  \widetilde{d\theta} \
i\varphi(\theta-\bar{\theta}+i\pi/2)(mL\sinh\theta-i\varphi(\theta-\bar\theta+i\pi/2))\\
&\qquad
+\sum_{n=2}^\infty \int \widetilde{d\theta_1}\dots
 \widetilde{d\theta_n}\
i\varphi(\theta_1-\bar{\theta}+i\pi/2)\varphi(\theta_1-\theta_2)\dots
\varphi(\theta_{n-1}-\theta_n) \times\\
&\hspace{6cm}\times (mL\sinh\theta_n-i\varphi(\theta_n-\bar\theta+i\pi/2)).
\end{split}
\end{equation}

The equations \eqref{LM1pkifejt}, \eqref{LM1pkifejt2} and
\eqref{n1_szepen_kifejezve} serve as an intermediate result for the
one-particle expectation values.
Expression \eqref{LM1pkifejt2} involves connected diagonal form factors with multiple insertions of
the same rapidity.
It will be shown in Section
\ref{sec:FF} that these cases
 can be expressed as sums of form factors
with only a single insertion of $\bar\theta+i\pi/2$. Moreover, after a resummation a remarkably
simple formula is found, which reproduces the
asymptotic result following from \eqref{fftcsa2-result2}.
   This
is presented in Section \ref{sec:1pfinal}.

\subsection{Two-particle states}

In the two-particle case the excited state TBA takes the form
\begin{equation}
  \label{2particleTBA}
   \eps(\theta)=mL\cosh\theta+
\log S(\theta-\bar\theta_1-i\pi/2)+\log S(\theta-\bar\theta_2-i\pi/2)
-\int  \frac{d\theta'}{2\pi}\varphi(\theta-\theta') \log(1+e^{-\eps(\theta')}).
\end{equation}

The integration contour $\mathcal{C}$ in \eqref{LMexc} consists of the
real line and two small circles around the points $\bar\theta_{1,2}+i\pi/2$. The
residues of the weight function are
\begin{equation*}
 \text{Res}_{\theta=\bar\theta_j+i\pi/2} \frac{1}{1+e^{\eps(\theta)}}=
\left(-\left.\frac{\partial\eps(\theta)}{\partial\theta}\right|_{\theta=\bar\theta_j+i\pi/2}\right)^{-1},
\qquad j=1,2.
\end{equation*}
Taking the derivatives we obtain
\begin{equation}
\label{n12kakki}
\begin{split}
&  n_j=-i
\left.\frac{\partial\eps(\theta)}{\partial\theta}\right|_{\theta=\bar\theta_j+i\pi/2}=
  mL\cosh\bar\theta_j+\varphi(\bar\theta_1-\bar\theta_2) +\varphi(0)-\\
&\hspace{4cm}-\int  \frac{d\theta'}{2\pi}i\varphi(\theta'-\bar\theta_j-i\pi/2)
\frac{1}{1+e^{\eps(\theta')}}
\frac{\partial\eps}{\partial\theta },\qquad j=1,2.
\end{split}
\end{equation}
It follows from \eqref{2particleTBA} that the derivative of the
pseudoenergy satisfies the integral equation
\begin{align}
\label{n12kakki2}
 (1-\hat{K})\frac{\partial\eps}{\partial\theta
 }=mL\sinh\theta+i\varphi(\theta-\bar\theta_1-i\pi/2)+i\varphi(\theta-\bar\theta_2-i\pi/2),
\end{align}
where the integral operator $K(\theta,\theta')$ is defined in
\eqref{Kkernel}. The explicit form of $n_{1,2}$ is
\begin{equation}
\begin{split}
\label{n1_szepen_kifejezve12}
& n_j =\varphi(0)+\varphi_{12}+
mL\cosh\bar\theta_j\\
&\qquad-\int  \widetilde{d\theta} \
i\varphi(\theta-\bar{\theta}_j-i\pi/2)
(mL\sinh\theta+i\varphi(\theta-\bar\theta_1-i\pi/2)+i\varphi(\theta-\bar\theta_2-i\pi/2))\\
&\qquad
-\sum_{n=2}^\infty \int \widetilde{d\theta_1}\dots
 \widetilde{d\theta_n}\
i\varphi(\theta_1-\bar{\theta}_j-i\pi/2)\varphi(\theta_1-\theta_2)\dots
\varphi(\theta_{n-1}-\theta_n) \times\\
&\hspace{4cm}\times 
(mL\sinh\theta_n+i\varphi(\theta_n-\bar\theta_1-i\pi/2)+i\varphi(\theta-\bar\theta_2-i\pi/2)),
\qquad j=1,2.
\end{split}
\end{equation}

Evaluating the residues in \eqref{LMexc} the result can be written in
the form 
\begin{equation}
\label{LM2pkifejt}
  \bra{\bar\theta_1,\bar\theta_2}\ordo
  \ket{\bar\theta_1,\bar\theta_2}_L=
\sum_{j,k,l} \mathcal{L}_{jkl},
\end{equation}
where $\mathcal{L}_{jkl}$ denotes those contributions of the $n=j+k+l$ term of
the original series where
$\bar\theta_1+i\pi/2$ has been inserted $j$ times,
$\bar\theta_2+i\pi/2$ has been inserted $k$ times, and there are $l$
auxiliary rapidities which are integrated over. Explicitly
\begin{equation}
\label{LM2pkifejt2}
  \mathcal{L}_{jkl}=\frac{1}{j!k!l!} \int \widetilde{d\theta_1}\dots \widetilde{d\theta_l}
\frac{F^\ordo_{2(j+k+l),c}(\bt_1+i\pi/2,\bt_1+i\pi/2,\dots,\bt_2+i\pi/2,\bt_2+i\pi/2,\dots,\theta_1,\dots,\theta_l)}
{n_1^jn_2^k}.
\end{equation}
Equations \eqref{LM2pkifejt}, \eqref{LM2pkifejt2} together with
\eqref{n12kakki} represent an explicit result for the two-particle
mean values. However, similar to the one-particle case further
investigation of the degenerate cases of the form factors 
leads to a remarkably simpler
formula, presented in Section
\ref{sec:2pfinal}. Eq. \eqref{LM2pkifejt2} motivates the study of
those diagonal form factors, where there are multiple insertions of
two different rapidities.


\section{Properties of the connected diagonal form factors}

\label{sec:FF}

In this section we consider the connected evaluation of the diagonal
form factors of a local operator $\mathcal{O}$. 
They are defined as
\begin{equation}
\label{fcdef}
    F^\mathcal{O}_{2n,c}(\theta_1\dots\theta_n)=
\text{finite piece of }\
F_{2n}^\mathcal{O}(\theta_1+i\pi+\eta_1,\dots,\theta_n+i\pi+\eta_n,\theta_n,\dots,\theta_1).
\end{equation}
It follows from the exchange axiom \eqref{exchange} that the functions above are
completely symmetric in their variables. 
They are meromorphic functions on the entire complex plain
and are invariant with respect to an overall boost of the
rapidities. It can be proven that they are $i\pi$ periodic and that
they inherit the clustering property: \cite{Delfino:1996nf}
  \begin{align}
    \label{eq:diagFF_cluster}
\lim_{ \Lambda\to\infty}
F^\mathcal{O}_{2n+2m,c}(\theta_1+\Lambda,\dots,\theta_n+\Lambda,\theta'_1,\dots,\theta'_m)=
\frac{1}{\vev{\mathcal{O}}} 
F^\mathcal{O}_{2n,c}(\theta_1,\dots,\theta_n)
F^\mathcal{O}_{2m,c}(\theta'_1,\dots,\theta'_m).
  \end{align}

No theorems are known about the growth properties of the diagonal form
factors. However we can assume that there exists a $K\in\valos^+$ such
that for any $n$
\begin{equation*}
\big|  F^\mathcal{O}_{2n,c}(\theta_1,\dots,\theta_n)\big|< n!
K^n,\qquad
\theta_j\in \complex.
\end{equation*}
This is in accordance with all previous experience and is enough to
ensure the convergence of the series \eqref{LM} for large enough $L$.

We note that the two-particle connected form factor
$F_{2c}^\ordo(\theta)$ does not depend on $\theta$ by
Lorentz-invariance. Therefore it will be simply denoted by
$F_{2c}^\ordo$ in the rest of this work.

\subsection{Degenerate cases}

\label{sec:degenerate}

In the following we consider degenerate diagonal form factors, ie. when there are multiple copies of
the same rapidity present. We only consider those
cases which are relevant for the one-particle and two-particle mean
values, namely when there are multiple copies of at most two different
rapidities. 

In the calculations we will extensively make use of sets and lists
of indices, therefore it is useful to introduce a few definitions and notations for
further use. 

A multiset is a generalization of a set allowing
members to appear more than once. The difference between a multiset
and a sequence is that in a multiset the order of the elements does
not matter, whereas in a sequence it does.
Multisets of
integers will be denoted using curly braces, for example
$A=\{a_1,a_2,\dots,a_n\}$. 
Sequences of integers will be denoted by braces:
  $s=(s_1,s_2,\dots,s_n)$.
If the explicit numbers are given then we don't use separation marks,
for example $s=(1342)$. Both for sequences and multisets, multiple addition of a given
number will be sometimes denoted in the superscript, for example
\begin{equation*}
  \{1^{(\times 3)}\}\equiv \{1,1,1\}.
\end{equation*}

Concatenation of sequences or addition of new elements will be denoted
simply by writing down the constituents without separation marks. For
example if 
\begin{equation*}
  A=(123) \qquad \text{and} \qquad B=(31)
\end{equation*}
then
\begin{equation*}
  (2A1B)=(2123131).
\end{equation*}

Unions of multisets is defined as a complete addition of the elements
and it will be denoted simply by comma, for example if $A=\{1,1\}$ and
$B=\{1,2,3\}$ then 
\begin{equation*}
  \{A,B\}=\{1,1,1,2,3\}.
\end{equation*}

The diagonal form factor evaluated on a multiset $A$ is defined as
\begin{equation*}
  F^\ordo_{2n,c}(A)\equiv F^\ordo_{2n,c}(\theta_{A_1},\theta_{A_2},\dots,\theta_{A_n}).
\end{equation*}

Given a sequence of integers $s$ we define
\begin{equation*}
  [s]_\varphi\equiv
\varphi_{s_1s_2}\varphi_{s_2s_3} \dots \varphi_{s_{n-1}s_n},
\end{equation*}
where it is understood that $n$ is the length of $s$ and
\begin{equation*}
  \varphi_{jk}=\varphi(\theta_j-\theta_k).
\end{equation*}
In the simplest case $[11]_\varphi=\varphi(0)$. Sometimes we will use
the shorthand $\varphi_0=\varphi(0)$. 

Given a multiset $A$ we define $S_{a,b}(A)$ to be the
set of all sequences, which contain every member of $A$ exactly once,
and where the first and
last elements are $a$ and $b$, respectively. It 
is understood that $S_{a,b}(A)$ is empty when $a\notin A$ or $b\notin
A$. For example 
\begin{equation*}
  S_{1,2}\big(\{1,1,2,3\}\big)=\{(1132),(1312)\}.
\end{equation*}

\subsubsection{Multiple copies of one rapidity}

The simplest degenerate case is when  in
the four-particle diagonal form factor the two rapidities
coincide. 
\begin{thm}
\begin{equation*}
  F^\ordo_{4c}(\theta,\theta)=2\varphi_0 F^\ordo_{2c}.
\end{equation*} 
\end{thm}
\begin{proof}
   Consider the form factor $F^\ordo_4(\theta_1,
 \theta_2,i\pi+\theta_3,i\pi+\theta_4)$ at $\theta_{1,2,3,4}\to
 \theta$. In this case there are 4 kinematical poles and no double singularities. Subtracting all
 four poles we obtain the fully
 connected form factor \footnote{The main idea of this proof was
   suggested  by G\'abor Tak\'acs.}
 \begin{equation*}
   \begin{split}
        F^\ordo_{4,fc}=&F^\ordo_4(\theta_1+i\pi,\theta_2+i\pi,\theta_3,\theta_4)-\\
&-\frac{i}{\theta_1-\theta_3} \left[ S(\theta_1-\theta_2)
  -S(\theta_3-\theta_4)\right] F^\ordo_2(\theta_2+i\pi,\theta_4)\\
&-\frac{i}{\theta_2-\theta_3}\left[1-S(\theta_3-\theta_4)S(\theta_1-\theta_2)\right] 
F^\ordo_2(\theta_1+i\pi,\theta_4)\\
&-\frac{i}{\theta_1-\theta_4} \left[ S(\theta_1-\theta_2)S(\theta_3-\theta_4)
  -1\right] F^\ordo_2(\theta_2+i\pi,\theta_3)\\
&-\frac{i}{\theta_2-\theta_4} \left[S(\theta_3-\theta_4)  -
  S(\theta_1-\theta_2)\right] F^\ordo_2(\theta_1+i\pi,\theta_3). 
   \end{split}
 \end{equation*}
The object above has manifestly the same exchange properties as the original
form factor. Therefore it vanishes as $\theta_1\to\theta_2$ or
$\theta_3\to\theta_4$ and even the point
$\theta_1=\theta_2=\theta_3=\theta_4$ is completely regular and
continuous. 
Using the approximation $S(\eps)\approx -1-i\varphi_0 \eps$ we obtain
 \begin{equation*}
   \begin{split}
        F^\ordo_{4,fc}\approx &F^\ordo_4(\theta_1+i\pi,\theta_2+i\pi,\theta_3,\theta_4)-\\
&-\frac{1}{\theta_1-\theta_3}
(\theta_1-\theta_2-\theta_3+\theta_4) \varphi_0 F^\ordo_2(\theta_2+i\pi,\theta_4)\\
&-\frac{1}{\theta_2-\theta_3}(\theta_3-\theta_4+\theta_1-\theta_2)
 \varphi_0 F^\ordo_2(\theta_1+i\pi,\theta_4)\\
&-\frac{1}{\theta_1-\theta_4} (-\theta_1+\theta_2-\theta_3+\theta_4)
 \varphi_0 F^\ordo_2(\theta_2+i\pi,\theta_3)\\
&-\frac{1}{\theta_2-\theta_4} (\theta_3-\theta_4  -\theta_1+\theta_2)
\varphi_0 F^\ordo_2(\theta_1+i\pi,\theta_3).
   \end{split}
 \end{equation*}
The connected FF is obtained by taking $\theta_1\to \theta_4$ and
$\theta_2\to \theta_3$ and subtracting the poles of the form
$(\theta_1-\theta_4)/(\theta_2-\theta_3)$ and
$(\theta_2-\theta_3)/(\theta_1-\theta_4)$:
\begin{equation*}
  \begin{split}
      F^\ordo_{4c}(\theta_1,\theta_2)=&\lim_{\theta_1\to
    \theta_4}\lim_{\theta_2\to \theta_3}\Big[
F^\ordo_4(\theta_1+i\pi,\theta_2+i\pi,\theta_3,\theta_4)- \\
& -\frac{\theta_1-\theta_4}{\theta_2-\theta_3}
 \varphi_0 F^\ordo_2(\theta_1+i\pi,\theta_4)
-\frac{\theta_2-\theta_3}{\theta_1-\theta_4}
 \varphi_0 F^\ordo_2(\theta_2+i\pi,\theta_3)\Big].
  \end{split}
\end{equation*}
Using 
\begin{equation*}
\lim_{\theta_{1,2,3,4}\to \theta} F^\ordo_{4,fc}(\theta_1+i\pi,\theta_2+i\pi,\theta_3,\theta_4)=0
\end{equation*}
we obtain
\begin{equation*}
  \lim_{\theta_1\to \theta_2} F^\ordo_{4c}(\theta_1,\theta_2)=
2\varphi_0 F^\ordo_{2,c}.
\end{equation*}
\end{proof}

A similar calculation can be performed in the three-particle case when
two rapidities coincide:
\begin{thm}
  \begin{equation*}
  F^\ordo_{6c}(\theta_1,\theta_1,\theta_3)=
2(\varphi_0 F^\ordo_{4c}(\theta_1,\theta_3) +\varphi_{13}^2 F^\ordo_{2c}).
\end{equation*}
\end{thm}
\begin{proof}
Consider the fully connected form factor
\begin{equation*}
\begin{split}
F^\ordo_{6,fc}=&F^\ordo_6(\theta_1+i\pi,\theta_2+i\pi,\theta_3+i\pi,\theta_3',\theta_2',\theta_1')-\\
&-\frac{i}{\theta_1-\theta_1'}\left[
S(\theta_1-\theta_2)S(\theta_1-\theta_3)S(\theta_2'-\theta_1')S(\theta_3'-\theta_1')-1
\right]    F^\ordo_4(\theta_2+i\pi,\theta_3+i\pi,\theta_3',\theta_2')\\
&-\frac{i}{\theta_1-\theta_2'}\left[
S(\theta_1-\theta_2)S(\theta_1-\theta_3)S(\theta_3'-\theta_2')-S(\theta_2'-\theta_1')
\right]    F^\ordo_4(\theta_2+i\pi,\theta_3+i\pi,\theta_3',\theta_1')\\
&-\frac{i}{\theta_2-\theta_1'}\left[
S(\theta_2-\theta_3)S(\theta_2'-\theta_1')S(\theta_3'-\theta_1')-
S(\theta_1-\theta_2)
\right]    F^\ordo_4(\theta_1+i\pi,\theta_3+i\pi,\theta_3',\theta_2')\\
&-\frac{i}{\theta_2-\theta_2'}\left[
S(\theta_2-\theta_3)S(\theta_3'-\theta_2')-S(\theta_1-\theta_2)S(\theta_2'-\theta_1')
\right]    F^\ordo_4(\theta_1+i\pi,\theta_3+i\pi,\theta_3',\theta_1').
\end{split}
\end{equation*}
This object satisfies the exchange axioms and therefore it vanishes at
the degenerate point.

Expanding the pre-factors to first order
\begin{equation*}
\begin{split}
F^\ordo_{6,fc}=&F^\ordo_6(\theta_1+i\pi,\theta_2+i\pi,\theta_3+i\pi,\theta_3',\theta_2',\theta_1')\\
&+\frac{1}{\eps_1}\left[
\varphi_0 (\eps_1-\eps_2) +\varphi_{13}(\eps_1-\eps_3)\right] 
   F^\ordo_4(\theta_2+i\pi,\theta_3+i\pi,\theta_3',\theta_2')\\
&-\frac{1}{\theta_1-\theta_2'}\left[\varphi_{13}(\theta_1-\theta_2'-\eps_3)+
\varphi_0(\theta_1+\theta_1'-\theta_2-\theta_2')
\right]    F^\ordo_4(\theta_2+i\pi,\theta_3+i\pi,\theta_3',\theta_1')\\
&-\frac{1}{\theta_2-\theta_1'}\left[\varphi_{13}(\theta_2-\theta_1'-\eps_3)
-\varphi_0(\theta_1+\theta_1'-\theta_2-\theta_2')
\right]    F^\ordo_4(\theta_1+i\pi,\theta_3+i\pi,\theta_3',\theta_2')\\
&+\frac{1}{\eps_2}\left[\varphi_{13}(\eps_2-\eps_3)-
\varphi_0 (\eps_1-\eps_2)
\right]    F^\ordo_4(\theta_1+i\pi,\theta_3+i\pi,\theta_3',\theta_1').
\end{split}
\end{equation*}
We have to collect all those contributions which are non-singular in
the limit $\eps_{1,2,3}\to 0$. From the first and last line we get
(already assuming $\theta_2\to\theta_1$)
\begin{equation*}
2  (\varphi_{0}+\varphi_{13})F^\ordo_{4c}(\theta_1,\theta_3).
\end{equation*}
In the second and third line we use
\begin{equation*}
  \begin{split}
     F^\ordo_4(\theta_2+i\pi,\theta_3+i\pi,\theta_3',\theta_1')&=
F^\ordo_{4c}(\theta_1,\theta_3)+
\left(\frac{\theta_2-\theta_1'}{\eps_3}+\frac{\eps_3}{\theta_2-\theta_1'}\right)\varphi_{13}F^\ordo_{2c}\\
 F^\ordo_4(\theta_1+i\pi,\theta_3+i\pi,\theta_3',\theta_2')&=
F^\ordo_{4c}(\theta_1,\theta_3)+
\left(\frac{\theta_1-\theta_2'}{\eps_3}+\frac{\eps_3}{\theta_1-\theta_2'}\right)\varphi_{13}F^\ordo_{2c}.
  \end{split}
\end{equation*}
This way we obtain the contributions
\begin{equation*}
  -\left(2\varphi_{13}+2\varphi_0-\varphi_0
\left(\frac{\theta_2-\theta_1'}{\theta_1-\theta_2'}+\frac{\theta_1-\theta_2'}{\theta_2-\theta_1'}
\right)
\right)F^\ordo_{4c}(\theta_1,\theta_3)+
\left(\frac{\theta_2-\theta_1'}{\theta_1-\theta_2'}+\frac{\theta_1-\theta_2'}{\theta_2-\theta_1'}
\right) \varphi_{13}^2 F^\ordo_{2c}.
\end{equation*}
Taking $\theta_1'\to\theta_1$ and $\theta_2'\to\theta_2$
\begin{equation*}
  -\left(2\varphi_{13}+4\varphi_0
\right)F^\ordo_{4c}(\theta_1,\theta_3)-2
 \varphi_{13}^2 F^\ordo_{2c}.
\end{equation*}
Adding all the contributions we obtain the statement of the theorem.
\end{proof}

The higher particle case of two coinciding rapidities is given as
follows.
\begin{thm}
\label{thm23m}
 Let $A=\{2,3,\dots,m\}$. Then
\begin{equation}
\label{deg1}
  F_{c}^\ordo(1,1,A)=2\mathop{\sum_{A=A^+\cup A^-}}
F^\ordo_{c}(1,A^+) \sum_{s\in S_{1,1}(1,1,A^-)} [s]_\varphi.
\end{equation}
The sum in \eqref{deg1} runs over all bi-partite
partitions of the set $A$. For example in the case of $A=\{2,3\}$ we have
\begin{equation*}
\begin{split}
  F^\ordo_{8c}(1,1,2,3)=2\Big(
&F^\ordo_{6c}(1,2,3)[11]_\varphi+\\
&F^\ordo_{4c}(1,2)[131]_\varphi+
F^\ordo_{4c}(1,3)[121]_\varphi+\\
&F^\ordo_{2c}(1) ([1231]_\varphi+[1321]_\varphi)
\Big).
\end{split}
\end{equation*}
\end{thm}
\begin{proof}
  The proof can be given following the same steps as in the three-particle case
  presented above.
\end{proof}

The following theorem concerns the case when there are more than two
occurrences of $\theta_1$ in the form factor:
\begin{thm}
\label{thmegyesek}
Let $B=\{1^{(\times n)}\}$ with $n\ge 2$ and $A$ an
arbitrary multiset not including $1$. Then
\begin{equation}
\label{egyesek}
F^\ordo_{2(n+m),c}(B,A)=
n! \mathop{\sum_{A=A^+\cup A^-}}
F^\ordo_{c}(1,A^+) 
\sum_{s\in S_{1,1}(B,A^-)} [s]_\varphi,
\end{equation}
where $|A|=m$.

\end{thm}
\begin{proof}
For technical reasons introduce new labels to the first $n$ rapidities
as $(1_a,1_b,\dots)$. The degenerate form factor can be obtained by
repeated use of eq. \eqref{deg1}. Here for technical reasons we
distinguish the first rapidities too, therefore \eqref{deg1} can be
written down without a factor of 2, extending the summation over all
possible paths starting and ending with either of the 1's. Then the
repeated use of Theorem \ref{thm23m} results in all
possible paths starting and ending with one of the 1's. Removing the
labels at the end results in a factor of $n!$.
\end{proof}

As a special case of the above theorem we obtain the result for
the completely degenerate case:
\begin{equation}
\label{teljesendegeneralt}
  F^\ordo_{2n,c}(1,1,\dots,1)=
 n!  \varphi(0)^{n-1}  F^\ordo_{2,c}.
\end{equation}

\subsubsection{Multiple copies of two rapidities}

\begin{thm}
\label{csak-kettesek}
Let $A=\{1^{(\times n)},2^{(\times m)}\}$ with $n,m\ge 2$.
Then
  \begin{equation*}
    F^\ordo_{2(n+m),c}(A)=n!m!\left(
\frac{F^\ordo_{4c}(1,2)}{\varphi_{12}}\sum_{s\in S_{1,2}(A)} [s]_\varphi+
F^\ordo_{2c}(1) \sum_{s\in S_{1,1}(A)} [s]_\varphi+
F^\ordo_{2c}(2)\sum_{s\in S_{2,2}(A)} [s]_\varphi\right).
  \end{equation*}
\end{thm}
\begin{proof}
  There are two ways to calculate this diagonal case. We can apply Theorem
  \ref{thmegyesek} first to the 1's, then to the 2's, or the other way
  around. Performing the steps in the first way, it is easy
  to see that the only contributions to the coefficient of $F^\ordo_{2c}(1)$
  are those given above. The coefficient of $F^\ordo_{2c}(2)$ can be
  obtained by performing the two steps in the second way. 

In order to obtain the coefficient of $F^\ordo_{4c}(1,2)$ 
we attach extra labels to the numbers as
  $\{1_a,1_b,\dots,2_a,2_b,\dots\}$. We obtain terms of the
  form $[1B1]_\varphi[2C2]_\varphi$, where
  $A=\{1,1,2,2,B,C\}$ and $C$ consists only of 2's ($B$ may include both
  1's and 2's) and terms of the form $[1B1]_\varphi$ such that
  $A=\{1,1,2,B\}$. 
For the terms in the first case 
we use the identity
  \begin{equation*}
    [1B1]_\varphi[2C2]_\varphi=\frac{[1B12C2]_\varphi}{\varphi_{12}},
  \end{equation*}
whereas in the second case we have
  \begin{equation*}
    [1B1]_\varphi=\frac{[1B12]_\varphi}{\varphi_{12}}.
  \end{equation*}
This way we obtain a summation over all paths starting with one of the
1's and ending with one of the 2's. Finally removing the labels we obtain a factor of $n!m!$.
\end{proof}

Finally we consider the case when there are a multiple  1's
and 2's and an arbitrary number of other rapidities. 
\begin{thm}
\label{egyesekkettesek}
  Let $A=\{1^{(\times n)},2^{(\times m)},B\}$ 
with $B=\{3,4,\dots,k+2\}$ such that $|B|=k$. The degenerate form factor
  is equal to
\begin{equation}
\begin{split}
F^\ordo_{2(n+m+o),c}(A)=&n!m!
\sum_{B=B^+\cup B^-}
\Big(
\frac{F^\ordo_c(1,2,B^+)}{\varphi_{12}} \sum_{s\in \tilde S_{1,2}(A\setminus B^+)} [s]_\varphi+\\
&\hspace{2cm}+ F^\ordo_c(1,B^+) \sum_{s\in S_{1,1}(A\setminus B^+)} [s]_\varphi+
F^\ordo_c(2,B^+) \sum_{s\in S_{2,2}(A\setminus B^+)} [s]_\varphi
\Big),
\end{split}
\end{equation}
where $\tilde S_{1,2}(A)$ is the sum of all paths consisting of the
elements of $A$, which start with 1, end with 2, and the first number
after the rightmost 1 is 2. For example 
\begin{equation*}
  \tilde S_{1,2}(\{1,1,2,2,3\})=\{
(12312),
(13122),
(13212),
(11232)
\}.
\end{equation*}
\end{thm}
\begin{proof}
 The theorem  can be proven along the lines of the previous proof. The
 coefficients of the form factors $F_c^\ordo(1,B^+)$ and $F^\ordo_c(2,B^+)$ are
 easily obtained by performing the steps of theorem \ref{thmegyesek}
 first for the 1's, then for the 2's, or the other way around,
 respectively.

In order to obtain the coefficients of the form factors $F^\ordo_c(1,2,B^+)$
we first perform theorem \ref{thmegyesek} for the 1's and then for the
2's. This way we obtain terms of the form $[1C1]_\varphi[2D2]_\varphi$
with $A\setminus B^+=(1,1,2,2,C,D)$ such that $D$ contains no 1, 
and terms of the form $[1C1]_\varphi$ with $A\setminus
B^+=(1,1,2,C)$. Using again the identities 
  \begin{equation*}
    [1C1]_\varphi[2D2]_\varphi=\frac{[1C12D2]_\varphi}{\varphi_{12}},
\qquad
    [1C1]_\varphi=\frac{[1C12]_\varphi}{\varphi_{12}}
  \end{equation*}
we obtain the desired statement.
\end{proof}

\section{One-particle expectation values}

\label{sec:1pfinal}

In this section we calculate the one-particle mean value
\eqref{LM1pkifejt}-\eqref{LM1pkifejt2} using the results of the
previous section for the degenerate form factors.

It follows from \eqref{fftcsa2-result2} that the asymptotic result is
\begin{equation}
\label{1p-asympt}
  \bra{\bar\theta}\ordo\ket{\bar\theta}=\frac{F^\ordo_{2,c}}{mL\cosh\bar\theta}+\vev{\ordo}.
\end{equation}
It is instructive to first obtain this asymptotic formula,
neglecting all exponential corrections.
This is presented in the following subsection. 
The exact
result with all terms included is calculated in subsection \ref{sec:1pb}.

\subsection{All orders in $1/L$}

In the asymptotic approximation we only keep the contributions with $k=0$
from 
\eqref{LM1pkifejt}-\eqref{LM1pkifejt2}. These are the terms which do not contain the weight
function
$1/(1+e^{\eps(\theta)})\sim e^{-mL}$: 
\begin{equation}
\label{1pasff}
\bra{\bar\theta}\mathcal{O}\ket{\bar\theta}=\vev{\mathcal{O}}+
 \sum_{j=1}^\infty \frac{1}{j!}
\frac{1}{n_1^j}
F^\ordo_{2j,c}(\bar\theta,\bar\theta,\dots,\bar\theta).
\end{equation}
In this approximation $n_1$ is
\begin{equation*}
n_1=mL\cosh\bar\theta+\varphi(0)+\ordo(e^{-mR}).
\end{equation*}
The relation \eqref{teljesendegeneralt}  is used to sum up
the second term in \eqref{1pasff} 
 as
\begin{equation}
\label{1pasf}
\begin{split}
&   \sum_{j=1}^\infty \frac{1}{j!}
\left( \frac{1}{mL\cosh\bar\theta+\varphi(0)}
 \right)^j
F^\ordo_{2j,c}(\bar\theta,\bar\theta,\dots,\bar\theta)=\\
&\hspace{3cm}=\frac{F^\ordo_{2,c}}{mL\cosh\bar\theta+\varphi(0)}
\sum_{j=1}^\infty  \left( \frac{\varphi(0)}{mL\cosh\bar\theta+\varphi(0)}
 \right)^{j-1}\\
&\hspace{3cm}=\frac{F^\ordo_{2,c}}{mL\cosh\bar\theta}.
\end{split}
\end{equation}
Putting together \eqref{1pasff} and \eqref{1pasf} we obtain indeed \eqref{1p-asympt}.
It is
interesting  that a resummation of an infinite number of terms was
required to obtain this simple formula.

\subsection{Exponential corrections}

\label{sec:1pb}

Here we consider all terms in
\eqref{LM1pkifejt}-\eqref{LM1pkifejt2}. The notations used in the
following calculation are given at the beginning of subsection \ref{sec:degenerate}.

A given $\mathcal{L}_{jk}$ can be
evaluated using the statement \eqref{egyesek} for the degenerate form factors as
\begin{equation*}
\begin{split}
&   \mathcal{L}_{jk}=\frac{1}{k!} \int \widetilde{d\theta_2}\dots \widetilde{d\theta_{k+1}}\
n_1^{-j}\times \\
&\hspace{1cm} 
\sum_{A_k=A^-\cup A^+}
F^\ordo_{2(1+|A^+|),c}(\bt+i\pi/2,A^+) 
\sum_{s\in S_{11}(1^{(\times j)},A^-)} [s]_\varphi,
\end{split}
\end{equation*}
where we assumed $j>1$, $A_k=\{2,3,\dots,k+1\}$ and we identified
$\theta_1=\bar\theta+i\pi/2$. 

Summing up all terms, grouping them according to the subset $A^+$ and using
the symmetry in the variables the mean value can be cast in the form
\begin{equation}
\label{egyfajtakifejtes}
    \bra{\bar\theta}\ordo \ket{\bar\theta}=
\vev{\ordo}_\eps+
\frac{\sum_{l=0}^\infty
\frac{1}{l!}\int \widetilde{d\theta_2}\dots \widetilde{d\theta_{l+1}} \
F^\ordo_{2(1+l),c}(\bt+i\pi/2,\theta_2,\dots,\theta_{l+1})}{N_1}
\end{equation}
with
\begin{equation}
\label{LM2}
    \vev{\ordo}_{\eps}=\sum_n \frac{1}{n!}
\int \frac{d\theta_1}{2\pi}\dots  \frac{d\theta_n}{2\pi}
\left(\prod_j \frac{1}{1+e^{\eps(\theta_j)}}\right)
F^\ordo_{2n,c}(\theta_1,\dots,\theta_n),
\end{equation}
and
\begin{equation}
\label{norma1}
  \frac{1}{N_1}=
\frac{1}{n_1}+\sum_{j=2}^\infty\sum_{l=0}^\infty
n_1^{-j} \frac{1}{l!}\int 
\widetilde{d\theta_2}\dots \widetilde{d\theta_{l+1}}\
\sum_{s\in S_{11}(1^{(\times j)},2,3,\dots,l+1)} [s]_\varphi.
\end{equation}
Here the first term $1/n_1$ comes from those contributions where there
is only a single insertion of $\bar\theta+i\pi/2$. The infinite sum
comes from the terms with multiple insertions.

On a
sequence $s$ we define the function
\begin{equation}
\label{segy}
  [s]_1=n_1^{-(m(1,s)-1)} \int \widetilde{d\theta_2}\dots \widetilde{d\theta}_{l_s-m(1,s)+1} \
[s]_\varphi.
\end{equation}
Here $l_s$ denotes the total length of the sequence and $m(1,s)$ denotes
the multiplicity of 1 in $s$. 
For example
\begin{equation*}
  [11231]_1=\frac{1}{n_1^2} \int \widetilde{d\theta_2} \widetilde{d\theta_3}\
  \varphi_{11} \varphi_{12}\varphi_{23}\varphi_{31}.
\end{equation*}
It is easy to see that under multiplication this function behaves as
\begin{equation}
\label{furaszorzat}
  [1A1]_1\times [1B1]_1=[1A1B1]_1,
\end{equation}
where $A$ and $B$ are arbitrary finite sequences. 

Let $P_{1,1}$ be the set of finite sequences with the following
properties: $s\in P_{1,1}$ if 
\begin{enumerate}
\item Either $s=(1^{(\times j)})$ with some $j\ge 2$ or
$s$ is a permutation of the sequence
\begin{equation*}
  (1^{(\times j)},2,3,\dots,k)
\end{equation*}
with some $j,k\ge 2$
\item $s$ starts and ends with 1.
\item When numbers other than 1 are present, they appear in increasing
  order.
\end{enumerate}
The first few examples are given below:
\begin{equation*}
  \begin{split}
P_{1,1}=\{&
(11),(111),(121),(1111),(1121),(1211),(1231),\\
&(11111),(12111),(11211),(11121),(11231),(12131),(12311),(12341),
\dots
\}.
  \end{split}
\end{equation*}

With these notations \eqref{norma1} can be written as
\begin{equation}
\label{norma2}
  \frac{n_1}{N_1}=1+\sum_{s\in P_{1,1}}  [s]_1.
\end{equation}
Note that the property 3 in the definition of $P_{1,1}$ was needed to
remove a factor of $1/l!$ from the expression \eqref{norma1}. 

Equation
\eqref{norma2} presents an explicit representation of $N_1$, however
it is possible to derive a simpler form. Let $T_{1,1}$ be the subset
of $P_{1,1}$ containing only those 
sequences in which 1 appears only twice, at the beginning and at the
end: 
\begin{equation*}
  T_{1,1}=\{(11),(121),(1231),(12341),\dots\}.
\end{equation*}
It can be considered the set of the ``connected'' components of $P_{1,1}$.

\begin{thm}
Inverting \eqref{norma1} gives
\begin{equation}
\label{allitas}
 \frac{N_1}{n_1}=1-\sum_{s\in T_{1,1}} [s]_1.
\end{equation}
\end{thm}
\begin{proof}
It needs to be proven that
\begin{equation*}
\sum_{s\in P_{1,1}} [s]_1- \sum_{s\in T_{1,1}} [s]_1=
\left(\sum_{s\in T_{1,1}} [s]_1\right) 
\left(\sum_{s\in P_{1,1}} [s]_1\right).
\end{equation*}
In this form the l.h.s. consists of the sum of all allowed sequences in which
the number 1 is present at least 3 times.
Given an arbitrary such sequence $s$ it is possible to
reconstruct it as $s=1A1B1$ such that $(1A1)\in T_{1,1}$ and $(1B1)\in
P_{1,1}$. This provides a bijection between the elements on the
l.h.s. and all possible products on the 
r.h.s. and using \eqref{furaszorzat} we obtain a complete equality of
the two sides.
\end{proof}

Explicitly expanding \eqref{allitas} we obtain
\begin{equation*}
\begin{split}
&  N_1=
n_1-\varphi(0)
-\int \widetilde{d\theta_2}\ \varphi(\theta_2-\bar\theta-i\pi/2)^2\\
&\hspace{1cm}-\int \widetilde{d\theta_2}\widetilde{d\theta_3}\
\varphi(\theta_2-\bar\theta-i\pi/2)\varphi(\theta_2-\theta_3)
\varphi(\theta_3-\bar\theta-i\pi/2)-\dots
\end{split}
\end{equation*}
Substituting the explicit representation \eqref{n1_szepen_kifejezve} 
for $n_1$ results in
\begin{equation}
\begin{split}
\label{norm_factor_szepen_kifejezve}
& 
N_1=
mL\Big(\cosh\bar\theta
+\int  \widetilde{d\theta} \
i\varphi(\theta-\bar{\theta}+i\pi/2)\sinh\theta\\
&\qquad
+\sum_{n=2}^\infty \int \widetilde{d\theta_1}\dots
 \widetilde{d\theta_n}\
i\varphi(\theta_1-\bar{\theta}+i\pi/2)\varphi(\theta_1-\theta_2)\dots
\varphi(\theta_{n-1}-\theta_n) \sinh\theta_n
\Big).
\end{split}
\end{equation}
This is our final formula for the normalization factor in \eqref{egyfajtakifejtes}.

\subsection{Interpretation of the result}

Consider the function $\bar Q(\bar \theta)$ defined in
\eqref{norm_fact}. Taking the total derivative with respect to
$\bar\theta$ 
\begin{equation}
\label{Q1def1}
\frac{d \bar Q}{d\bar\theta}=
mL\cosh\bar\theta
+\int  \frac{d\theta'}{2\pi}i\varphi(\theta'-\bar{\theta}+i\pi/2)
\frac{1}{1+e^{\eps(\theta')}}
\left(\frac{\partial\eps}{\partial\theta
  }+\frac{\partial\eps}{\partial\bar \theta }\right).
\end{equation}
Here it is understood that the pseudoenergy function $\eps(\theta)$
also depends on $\bar\theta$ through the source term in \eqref{eq:excited_eps_sinhG}.

Differentiating \eqref{eq:excited_eps_BY} 
\begin{align}
\label{dddd}
 (1-\hat{K})
\left(\frac{\partial\eps}{\partial\theta }+\frac{\partial\eps}{\partial\bar\theta }\right)
=mL\sinh\theta
\end{align}
with $K(\theta,\theta')=\varphi(\theta-\theta')$. Inverting the
integral equation and substituting into \eqref{Q1def1}
\begin{equation}
\label{Qdef2}
\begin{split}
 \frac{\partial \bar Q}{\partial \bar\theta}=&
mL\Big(\cosh\bar\theta
+\int  \widetilde{d\theta} \
i\varphi(\theta-\bar{\theta}+i\pi/2)\sinh\theta\\
&
+\sum_{n=2}^\infty \int \widetilde{d\theta_1}\dots
 \widetilde{d\theta_n}\
i\varphi(\theta_1-\bar{\theta}+i\pi/2)\varphi(\theta_1-\theta_2)\dots
\varphi(\theta_{n-1}-\theta_n) \sinh\theta_n
\Big)=\\
=&N_1.
\end{split}
\end{equation}
In other words, the normalization factor obtained in the previous
subsection coincides with the exact total derivative of the
quantization condition for the rapidity $\bar\theta$. This is
analogous to the asymptotic result \eqref{fftcsa2-result2}, which
involves the derivatives of the (asymptotic) Bethe equations.

To conclude this section we present the one-particle result in a
compact formula:
\begin{equation}  
\label{LM1p}
\begin{split}
&    \bra{\bar\theta}\ordo \ket{\bar\theta}_L=\\
&\sum_{n=0}^\infty \frac{1}{n!}
\int \frac{d\theta_1}{2\pi}\dots  \frac{d\theta_n}{2\pi}
\left(\prod_j \frac{1}{1+e^{\eps(\theta_j)}}\right)
F^\ordo_{2n,c}(\theta_1,\dots,\theta_n)\\
&+\left( \frac{\partial \bar Q}{\partial \bar\theta}\right)^{-1}
\sum_{n=0}^\infty \frac{1}{n!}
\int \frac{d\theta_1}{2\pi}\dots  \frac{d\theta_n}{2\pi}
\left(\prod_j \frac{1}{1+e^{\eps(\theta_j)}}\right)
F^\ordo_{2(n+1),c}(\bar\theta+i\pi/2,\theta_1,\dots,\theta_n).
\end{split}
\end{equation}
The infinite series in the second line above is a result of the terms $\mathcal{L}_{jk}$ with $j=0$
and it has exactly the same form as the original LeClair-Mussardo
series; the only difference is that the excited state pseudoenergy
function is used in the weight functions. The third line is an
analogous series, with the two differences being the presence of the
rapidity $\bar\theta+i\pi/2$ and the normalization factor which is of order $1/L$.

\section{Two-particle expectation values}

\label{sec:2pfinal}

Here we perform the partial summation of the expansion
\eqref{LM2pkifejt}-\eqref{LM2pkifejt2} for the two-particle mean-value. The asymptotic result
following from \eqref{fftcsa2-result2} is
\begin{equation}
  \label{2p-asympt}
\begin{split}
&\bra{\bar\theta_1,\bar\theta_2}\ordo
\ket{\bar\theta_1,\bar\theta_2}_L=\vev{\ordo}+\\
&\hspace{2cm}+\frac{F_{4c}^\ordo(\bar\theta_1,\bar\theta_2)
+F_{4c}^\ordo(\bar\theta_1) (mL\cosh\bar\theta_2+\varphi_{12})+
F_{4c}^\ordo(\bar\theta_2) (mL\cosh\bar\theta_1+\varphi_{12})
}{\rho_2(\bar\theta_1,\bar\theta_2)},
\end{split}
\end{equation}
with
\begin{equation*}
  \rho_2(\bar\theta_1,\bar\theta_2)=m^2L^2
  \cosh\bar\theta_1\cosh\bar\theta_2+
mL(\cosh\bar\theta_1+\cosh\bar\theta_2) \varphi_{12}.
\end{equation*}
The summation procedure to obtain this asymptotic result is
already involved and is presented in the following subsection. 
The evaluation of the exact result
\eqref{LM2pkifejt}-\eqref{LM2pkifejt2} with all exponential
corrections included is presented in subsection \ref{sec:exp2p}.

\subsection{All orders in $1/L$}

Dropping the exponential corrections amounts to summing up the terms
with $l=0$ in \eqref{LM2pkifejt}. This leads to
\begin{equation}
\label{2pa}
\begin{split}
&\bra{\bar\theta_1,\bar\theta_2}\mathcal{O}\ket{\bar\theta_1,\bar\theta_2}=
\sum_{j,k=1}^\infty \mathcal{L}_{jk0}=\\
&\hspace{3cm}=\vev{\mathcal{O}}+
 \sum_{m=1}^\infty \frac{1}{m!}
\sum_{i_1\dots i_m=1}^2
\frac{F^\ordo_{2m,c}(\bar\theta_{i_1},\bar\theta_{i_2},\dots,\bar\theta_{i_m})}
{n_{i_1}n_{i_2}\dots n_{i_m}},
\end{split}
\end{equation}
where the normalization factors in this approximation are
\begin{equation*}
\begin{split}
  \frac{1}{n_j}&\equiv 
\text{Res}_{\theta=\bar\theta_j} \frac{1}{1+e^{\eps(\theta)}}=
\frac{1}{mL\cosh\bar\theta_j+\varphi(0)+\varphi(\bar\theta_1-\bar\theta_2)}+
\ordo(e^{-mR}),\qquad\qquad j=1,2.
\end{split}
\end{equation*}
Let us denote by $\mathcal{L}_j$ the sum of the terms which include $j$
rapidities in total:
\begin{equation*}
  \mathcal{L}_j=\sum_{k+l=j} \mathcal{L}_{kl0}.
\end{equation*}
The first few cases can be evaluated using theorem \ref{csak-kettesek}
explicitly as
\begin{equation*}
\begin{split}
  \mathcal{L}_1=&F^\ordo_{2c} \left(\frac{1}{n_1}+\frac{1}{n_2}\right)\\
\mathcal{L}_2=&F^\ordo_{2c}
\left(\frac{\varphi_0}{n_1^2}+\frac{\varphi_0}{n_2^2}\right)+F^\ordo_{4c}(1,2) \frac{1}{n_1n_2}\\
\mathcal{L}_3=&F^\ordo_{2c}
\left(\frac{\varphi_0^2}{n_1^3}+\frac{\varphi_0^2}{n_2^3}+
\frac{\varphi_{12}^2}{n_1^2n_2}+\frac{\varphi_{12}^2}{n_2^2n_1}
\right)\\
&+F^\ordo_{4c}(1,2)\left(\frac{\varphi_0}{n_1^2n_2}+\frac{\varphi_0}{n_2^2n_1}
\right)\\
\mathcal{L}_4=&
F^\ordo_{2c}
\left(\frac{\varphi_0^3}{n_1^4}+\frac{\varphi_0^3}{n_2^4}+
\frac{2\varphi_{12}^2\varphi_0}{n_1^3n_2}+\frac{2\varphi_{12}^2\varphi_0}{n_2^3n_1}+
\frac{2\varphi_{12}^2\varphi_0}{n_2^2n_1^2}
\right)\\
&+
F^\ordo_{4c}(1,2)\left(\frac{\varphi_0^2}{n_1^3n_2}+\frac{\varphi_0^2}{n_2^3n_1}+
\frac{\varphi_0^2+\varphi_{12}^2}{n_1^2n_2^2}
\right).
\end{split}
\end{equation*}

All terms in this series 
can be obtained by an auxiliary problem. Let $A_i$, $i=0\dots\infty$ a series of $2\times 2$
matrices satisfying the recursion relation
\begin{equation}
\label{recursion1}
  A_{i+1}=\begin{pmatrix}
  \frac{1}{n_1} & 0 \\ 0 & \frac{1}{n_2}
\end{pmatrix}+
  \begin{pmatrix}
\frac{\varphi_0}{n_1} & \frac{\varphi_{12}}{n_1}    \\
\frac{\varphi_{12}}{n_2} & \frac{\varphi_0}{n_2}
  \end{pmatrix} A_{i},\qquad\qquad
A_0=0.
\end{equation}

\begin{thm}
  The summation of the series $\mathcal{L}_j$ to a given order $k$ is given by
  \begin{equation}
\label{koztes2p}
    \sum_{j=0}^k \mathcal{L}_j=\text{Tr} \left[
      \begin{pmatrix}
        F^\ordo_{2c} & \frac{F^\ordo_{4c}(1,2)}{2\varphi_{12}}\\
\frac{F^\ordo_{4c}(1,2)}{2\varphi_{12}} & F^\ordo_{2c}
      \end{pmatrix}
A_k
\right].
  \end{equation}
\end{thm}
\begin{proof}
The matrix multiplication in \eqref{recursion1} generates the sum over
all possible paths consisting of the numbers 1 and 2. It also generates the
proper normalization factors which appear in \eqref{2pa}.
The statement of the theorem then follows from Theorem
\eqref{csak-kettesek}. 
\end{proof}

Given that $L$ is large enough the expressions in \eqref{koztes2p} are
convergent and
\begin{equation}
\label{Aii}
    \sum_{j=0}^\infty \mathcal{L}_j=\text{Tr} \left[
      \begin{pmatrix}
        F^\ordo_{2c} & \frac{F^\ordo_{4c}(1,2)}{2\varphi_{12}}\\
\frac{F^\ordo_{4c}(1,2)}{2\varphi_{12}} & F^\ordo_{2c}
      \end{pmatrix}
A_\infty
\right],\quad \text{with}\quad
A_\infty=\lim_{j\to\infty} A_j.
\end{equation}

The limiting matrix $A_\infty$ can be obtained from \eqref{recursion1} as
\begin{equation}
\label{Ainfty}
  A_{\infty}=\left(I- \begin{pmatrix}
\frac{\varphi_0}{n_1} & \frac{\varphi_{12}}{n_1}    \\
\frac{\varphi_{12}}{n_2} & \frac{\varphi_0}{n_2}
  \end{pmatrix}\right)^{-1}
\begin{pmatrix}
  \frac{1}{n_1} & 0 \\ 0 & \frac{1}{n_2}
\end{pmatrix}
=\frac{1}{\rho_{2}}
\begin{pmatrix}
mL\cosh\bar\theta_2+\varphi_{12} & \varphi_{12} \\
\varphi_{12} & mL\cosh\bar\theta_1+\varphi_{12}
\end{pmatrix}.
\end{equation}
Putting together \eqref{Aii} and \eqref{Ainfty} leads to the desired
statement \eqref{2p-asympt}.

\subsection{Exponential corrections}

\label{sec:exp2p}

Here we use theorem \ref{egyesekkettesek} to evaluate the expansion
\eqref{LM2pkifejt}-\eqref{LM2pkifejt2}. The terms will be grouped
according to the indices $j$ and $k$.

The terms with no insertions of the external rapidities are $\mathcal{L}_{0,0,l}$
with $l=0\dots\infty$:
\begin{equation}
\label{LM3}
\sum_{l=0}^\infty \mathcal{L}_{0,0,l} \equiv   \vev{\ordo}_{\eps}=\sum_n \frac{1}{n!}
\int \frac{d\theta_1}{2\pi}\dots  \frac{d\theta_n}{2\pi}
\left(\prod_j \frac{1}{1+e^{\eps(\theta_j)}}\right)
F^\ordo_{2n,c}(\theta_1,\dots,\theta_n).
\end{equation}
The terms $\mathcal{L}_{1,0,l}$, $\mathcal{L}_{0,1,l}$ and $\mathcal{L}_{1,1,l}$ are simply
\begin{equation*}
  \begin{split}
\mathcal{L}_{1,0,l}&=  \frac{1}{n_1}  \frac{1}{l!} \int \widetilde{d\theta_3}\dots \widetilde{d\theta_{l+2}}\
F^\ordo_{2(1+l),c}(\bar\theta_1+i\pi/2,\theta_3,\dots,\theta_{l+2})\\
\mathcal{L}_{0,1,l}&=  \frac{1}{n_2}  \frac{1}{l!} \int
\widetilde{d\theta_3}\dots 
\widetilde{d\theta_{l+2}}\
F^\ordo_{2(1+l),c}(\bar\theta_2+i\pi/2,\theta_3,\dots,\theta_{l+2})\\
\mathcal{L}_{1,1,l}&=  \frac{1}{n_1n_2}  \frac{1}{l!} \int \widetilde{d\theta_3}\dots \widetilde{d\theta_{l+2}}\
F^\ordo_{2(2+l),c}(\bt_1+i\pi/2,\bar\theta_2+i\pi/2,\theta_3,\dots,\theta_{l+2}).
  \end{split}
\end{equation*}

The terms $\mathcal{L}_{j,0,l}$ and $\mathcal{L}_{0,j,l}$ with $j\ge 2$, ie. those with
multiple insertions of either one of 
$\bar\theta_{1,2}+i\pi/2$ can be evaluated using Theorem \ref{thmegyesek}:
\begin{equation*}
\begin{split}
&\mathcal{L}_{j,0,l}
   =\frac{1}{k!} \int \widetilde{d\theta_3}\dots \widetilde{d\theta_{k+2}}\
n_1^{-j}\times \\
&\hspace{1cm} 
\sum_{A_k=A^-\cup A^+}
F^\ordo_{2(1+|A^+|),c}(\bt_1+i\pi/2,A^+) 
\sum_{s\in S_{11}(1^{(\times j)},A^-)} [s]_\varphi,
\end{split}
\end{equation*}
where we assumed $j>1$, $A_k=\{3,\dots,k+2\}$ and in the notation for
the sequences we identified
$\theta_1\equiv \bar\theta_1+i\pi/2$. 

Similarly
\begin{equation*}
\begin{split}
&\mathcal{L}_{0,j,l}=\frac{1}{k!} \int \widetilde{d\theta_3}\dots \widetilde{d\theta_{k+2}}\
n_2^{-j}\times \\
&\hspace{1cm} 
\sum_{A_k=A^-\cup A^+}
F^\ordo_{2(1+|A^+|),c}(\bt_2+i\pi/2,A^+) 
\sum_{s\in S_{22}(2^{(\times j)},A^-)} [s]_\varphi,
\end{split}
\end{equation*}
with $\theta_2\equiv \bar\theta_2+i\pi/2$. 

The terms $\mathcal{L}_{j,1,l}$ can also be evaluated using Theorem
\ref{thmegyesek}:
\begin{equation*}
\begin{split}
&\mathcal{L}_{j,1,l}
   =\frac{1}{k!} \int \widetilde{d\theta_3}\dots \widetilde{d\theta_{k+2}}\
\frac{1}{n_1^{j}n_2}\times \\
&\hspace{1cm} 
\sum_{A_k=A^-\cup A^+}\Big(
F^\ordo_{2(2+|A^+|),c}(\bt_1+i\pi/2,\bt_2+i\pi/2,A^+) 
\sum_{s\in S_{11}(1^{(\times j)},A^-)} [s]_\varphi\\
&\hspace{3cm}+F^\ordo_{2(1+|A^+|),c}(\bt_1+i\pi/2,A^+) 
\sum_{s\in S_{11}(1^{(\times j)},2,A^-)} [s]_\varphi
\Big),
\end{split}
\end{equation*}
where again $A_k=\{3,\dots,k+2\}$. An analogous expression holds for
the terms $\mathcal{L}_{1,j,l}$ with the role of the numbers 1 and 2 interchanged.

Finally Theorem \ref{egyesekkettesek} yields the terms with multiple
insertions of both external rapidities:
\begin{equation*}
\begin{split}
&\mathcal{L}_{j,k,l}
   =\frac{1}{l!} \int \widetilde{d\theta_3}\dots \widetilde{d\theta_{l+2}}\
\frac{1}{n_1^{j}n_2^k}\times \\
&\hspace{1cm} 
\sum_{A_l=A^-\cup A^+}\Big(
F^\ordo_{2(1+|A^+|),c}(\bt_1+i\pi/2,A^+) 
\sum_{s\in S_{11}(1^{(\times j)},2^{(\times k)},A^-)} [s]_\varphi\\
&\hspace{3cm}+F^\ordo_{2(1+|A^+|),c}(\bt_2+i\pi/2,A^+) 
\sum_{s\in S_{22}(1^{(\times j)},2^{(\times k)},A^-)} [s]_\varphi\\
&\hspace{3cm}+\frac{1}{\varphi_{12}}F^\ordo_{2(2+|A^+|),c}(\bt_1+i\pi/2,\bt_2+i\pi/2,A^+) 
\sum_{s\in \tilde S_{12}(1^{(\times j)},2^{(\times k)},A^-)} [s]_\varphi
\Big),
\end{split}
\end{equation*}
with $A_l=\{3,\dots,l+2\}$.

Adding all the terms above and grouping them according to the form factor
content we obtain the expansion
\begin{equation}
\label{2p-kifejt}
  \begin{split}
&  \bra{\bar\theta_1,\bar\theta_2}\ordo
  \ket{\bar\theta_1,\bar\theta_2}_L=
\vev{\ordo}_\eps+\\
&\hspace{2cm}+\frac{\sum_{l=0}^\infty
\frac{1}{k!}\int \widetilde{d\theta_3}\dots \widetilde{d\theta_{k+2}} \
F^\ordo_{2(1+k),c}(\bt_1+i\pi/2,\theta_3,\dots,\theta_{k+2})}{N_1}\\
&\hspace{2cm}+
\frac{\sum_{k=0}^\infty
\frac{1}{k!}\int \widetilde{d\theta_3}\dots \widetilde{d\theta_{k+2}} \
F^\ordo_{2(1+k),c}(\bt_2+i\pi/2,\theta_3,\dots,\theta_{k+2})}{N_2}\\
&\hspace{2cm}+
\frac{\sum_{k=0}^\infty
\frac{1}{k!}\int \widetilde{d\theta_3}\dots \widetilde{d\theta_{k+2}} \
F^\ordo_{2(2+k),c}(\bt_1+i\pi/2,\bt_2+i\pi/2,\theta_3,\dots,\theta_{k+2})}{N_{12}},
  \end{split}
\end{equation}
where the normalization factors are 
\begin{equation}
\label{norma121}
\begin{split}
  \frac{1}{N_1}=
&\frac{1}{n_1}+\sum_{j=2}^\infty\sum_{l=0}^\infty
\frac{1}{n_1^{j}} \frac{1}{l!}\int 
\widetilde{d\theta_3}\dots \widetilde{d\theta_{l+2}}\
\sum_{s\in S_{11}(1^{(\times j)},3,\dots,l+2)} [s]_\varphi\\
&+\sum_{j=2}^\infty\sum_{k=1}^\infty \sum_{l=0}^\infty
\frac{1}{n_1^{j}n_2^k} \frac{1}{l!}\int 
\widetilde{d\theta_3}\dots \widetilde{d\theta_{l+2}}\
\sum_{s\in S_{11}(1^{(\times j)},2^{(\times k)},3,\dots,l+2)} [s]_\varphi,
\end{split}
\end{equation}
\begin{equation}
\label{norma122}
\begin{split}
  \frac{1}{N_2}=
&\frac{1}{n_2}+\sum_{j=2}^\infty\sum_{l=0}^\infty
\frac{1}{n_2^{j}} \frac{1}{l!}\int 
\widetilde{d\theta_3}\dots \widetilde{d\theta_{l+2}}\
\sum_{s\in S_{22}(2^{(\times j)},3,\dots,l+2)} [s]_\varphi\\
&+\sum_{j=2}^\infty\sum_{k=1}^\infty \sum_{l=0}^\infty
\frac{1}{n_2^{j}n_1^k} \frac{1}{l!}\int 
\widetilde{d\theta_3}\dots \widetilde{d\theta_{l+2}}\
\sum_{s\in S_{22}(1^{(\times k)},2^{(\times j)},3,\dots,l+2)} [s]_\varphi
\end{split}
\end{equation}
and finally
\begin{equation}
\label{norma1212}
\begin{split}
  \frac{1}{N_{12}}=
&\frac{1}{n_1n_2}
+\sum_{j=2}^\infty\sum_{l=0}^\infty
\frac{1}{n_1^{j}n_2} \frac{1}{l!}\int 
\widetilde{d\theta_3}\dots \widetilde{d\theta_{l+2}}\
\sum_{s\in S_{11}(1^{(\times j)},3,\dots,l+2)} [s]_\varphi\\
&+\sum_{j=2}^\infty\sum_{l=0}^\infty
\frac{1}{n_1n_2^j} \frac{1}{l!}\int 
\widetilde{d\theta_3}\dots \widetilde{d\theta_{l+2}}\
\sum_{s\in S_{22}(2^{(\times j)},3,\dots,l+2)} [s]_\varphi\\
&+\frac{1}{\varphi_{12}}\sum_{j=2}^\infty\sum_{k=2}^\infty \sum_{l=0}^\infty
\frac{1}{n_1^{j}n_2^k} \frac{1}{l!}\int 
\widetilde{d\theta_3}\dots \widetilde{d\theta_{l+2}}\
\sum_{s\in \tilde S_{12}(1^{(\times j)},2^{(\times k)},3,\dots,l+2)} [s]_\varphi.
\end{split}
\end{equation}

The double and triple sums in the previous three formulas can be
simplified dramatically. However, first we have to introduce new
notations motivated by the expressions above.

Let $R_{a,b}$, $a,b=1,2$  be four sets of finite sequences with the following
properties: $s\in R_{a,b}$ if 
\begin{enumerate}
\item Either $s=(1^{(\times j)})$ or $s=(2^{(\times j)})$ with some $j\ge 2$, or
$s$ is a permutation of the sequences
\begin{equation*}
  (1^{(\times j)},2^{(\times k)})
\end{equation*}
with some $j\ge 1$, $k\ge 1$, or 
\begin{equation*}
  (1^{(\times j)},2^{(\times k)},3,\dots,l)
\end{equation*}
with some $j+k \ge 2$ and $l\ge 3$.
\item $s$ starts with $a$ and ends with $b$.
\item When numbers other than 1  and 2 are present, they appear in
  increasing order.
\end{enumerate}
The difference between the sets $R_{a,b}$ and the set $P_{1,1}$ introduced in
section \ref{sec:1pb} is that in the sequences in  $R_{a,b}$ the number 2 can
appear multiple times at arbitrary positions and only the ordering of
the numbers greater than 2 is constrained.

It is useful to introduce a shorthand for the integrals entering
\eqref{norma121}-\eqref{norma1212}. We define
\begin{equation}
\label{ezdenagyfaszsag}
  [s]_{12}=n_1^{-m(1,s)}n_2^{-m(2,s)}n_{s(l_s)} 
\int \widetilde{d\theta_3}\dots \widetilde{d\theta}_{l_s-l(1,s)-l(2,s)}\    [s]_\varphi,
\end{equation}
where $m(1,s)$ and $m(2,s)$ denote the multiplicity of the numbers 1
and 2 in $s$, respectively and $l_s$ is simply the length of $s$
(therefore $s(l_s)$ denotes the last element of the sequence).  
For example 
\begin{equation*}
  [21321]_{12}=\frac{1}{n_1n_2^2} \varphi_{12}^2 
\int \widetilde{d\theta_3}\  \varphi_{23}\varphi_{13}.
\end{equation*}
This function satisfies
\begin{equation}
\label{szorzo-tulajdonsag}
  [aAb]_{12}\times [bBc]_{12}=[aAbBc]_{12}
\end{equation}
for every $a,b,c=1,2$ and $A,B$ being arbitrary finite sequences.

With these notations the first two normalization factors 
in \eqref{2p-kifejt}
can be
expressed as
\begin{equation}
  \label{norma1b}
    \frac{1}{N_1}=
\frac{1}{n_1}+\frac{1}{n_1}
\sum_{s\in R_{1,1}}  [s]_{12}
\end{equation}
\begin{equation}
  \label{norma1c}
    \frac{1}{N_2}=
\frac{1}{n_2}+\frac{1}{n_2}
\sum_{s\in R_{2,2}}  [s]_{12}.
\end{equation}

It is also possible to obtain a compact formula for $N_{12}$ expressed
as a simple sum. We define $\tilde R_{1,2}$ to be the subset of $R_{1,2}$ which
includes sequences where the first number after the rightmost 1 is
2. The first few sequences are 
\begin{equation*}
  \begin{split}
\tilde R_{1,2}=\Big\{
(12),(112),(122),
(1112),(1122),(1222),(1232),(1312),\dots
\Big\}.
  \end{split}
\end{equation*}

\begin{thm}
\begin{equation}
  \label{norma1d}
      \frac{1}{N_{12}}=
\frac{1}{\varphi_{12}}\frac{1}{n_2}
\sum_{s\in \tilde R_{1,2}}  [s]_{12}.
\end{equation}  
\end{thm}
\begin{proof}
The idea of the proof is to show that each term in \eqref{norma1212}
corresponds to one of the elements of $\tilde R_{1,2}$ and that
\eqref{norma1d} provides the correct normalization factors.

The first term in \eqref{norma1212} corresponds to the sequence
$(12)\in  \tilde R_{1,2}$. We have by definition
\begin{equation*}
  \frac{1}{n_1n_2}=\frac{1}{\varphi_{12}} \frac{[12]_\varphi}{n_1n_2}.
\end{equation*}

Concerning the second term in the r.h.s. of \eqref{norma1212} note
that if $s\in S_{11}(1^{(\times j)},3,\dots,l+2)$ then the
 sequence $(s2)$ satisfies
\begin{equation*}
(s2)\in \tilde S_{12}(1^{(\times j)},2,3,\dots,l+2), \qquad
 [s2]_\varphi={\varphi_{12}}[s]_\varphi.
\end{equation*}
Similarly, if  $s\in S_{22}(2^{(\times
  j)},3,\dots,l+2)$ then the sequence $(1s)$ satisfies
\begin{equation*}
(1s)\in \tilde S_{12}(1,2^{(\times j)},3,\dots,l+2), \qquad
 [1s]_\varphi={\varphi_{12}}[s]_\varphi.
\end{equation*}
These two cases give the sequences of $\tilde R_{1,2}$ in which the
number 2 (or 1) appears only once, respectively.

Finally the the last line of \eqref{norma1212} yields all the
remaining sequences, where both 1 and 2 appear at least twice.

Using the definition \eqref{ezdenagyfaszsag} we obtain the statement of the theorem.
\end{proof}

We introduce the four sets of sequences $T_{a,b}$ with $a,b=1,2$:
\begin{equation}
  \label{Tab}
 T_{a,b}\equiv \{(ab),(a3b),(a34b),(a345b),
\dots\}.
\end{equation}

Consider the auxiliary problem
\begin{equation}
\label{recursion2}
  A_{i+1}=\begin{pmatrix}
  \frac{1}{n_1} & 0 \\ 0 & \frac{1}{n_2}
\end{pmatrix}+M\times
  A_{i}\qquad\qquad
A_0=0
\end{equation}
with
\begin{equation*}
M=  \begin{pmatrix}
\sum_{s\in T_{1,1}} [s]_{12} & \sum_{s\in T_{1,2}} [s]_{12} \\
\sum_{s\in T_{2,1}} [s]_{12} & \sum_{s\in T_{2,2}} [s]_{12}  
  \end{pmatrix}.
\end{equation*}

This recursion relation generates all sequences of
integers which are elements of the sets $R_{a,b}$. 

\begin{thm}
  The limiting matrix $A_\infty$ has the following diagonal elements:
  \begin{equation}
\label{Alim1}
      (A_\infty)_{1,1}=\frac{1}{N_1} \qquad
    (A_\infty)_{2,2}=\frac{1}{N_2}.
  \end{equation}
\end{thm}
\begin{proof}
The statement of the theorem follows from the property \eqref{szorzo-tulajdonsag}.
\end{proof}

\begin{thm}
  The off-diagonal element $(A_\infty)_{1,2}$ is
  \begin{equation}
\label{furcsa-allitas}
    (A_\infty)_{1,2}=
  \frac{n_1}{n_2}\frac{1}{\varphi_{12}}
 \sum_{s\in T_{1,2}} [s]_{12} 
\sum_{s\in \tilde R_{1,2}}  [s]_{12}.
  \end{equation}
\end{thm}
\begin{proof}
The iterative procedure gives
\begin{equation*}
     (A_\infty)_{1,2}=
\frac{1}{n_2}\sum_{s\in R_{1,2}} [s]_{12}.
\end{equation*}
A given sequence $s\in R_{1,2}$ can be written as
\begin{equation*}
  s=(A1C2B),
\end{equation*}
such that if $A$ is not empty then it starts with 1, if $B$ is not
empty then it ends with 2, and neither $B$ nor $C$ contain any
1's. Using the identity
\begin{equation*}
 [A1C2B]_{12}=\frac{n_1}{\varphi_{12}} [A12B]_{12} [1C2]_{12}
\end{equation*}
we obtain an exact bijection between the two sides of \eqref{furcsa-allitas}.
\end{proof}

We introduce the quantities
\begin{equation*}
\mathcal{N}_\varphi\equiv n_1 \sum_{s\in T_{1,2}} [s]_{12}  
\end{equation*}
and
\begin{equation}
\label{Njj}
\mathcal{N}_j\equiv n_j(1-M_{j,j})-\mathcal{N}_\varphi
\qquad\qquad
j=1,2.
\end{equation}

The explicit form of $\mathcal{N}_\varphi$ is
\begin{equation}
\begin{split}
\label{norm_factor_szepen_kifejezve2p_phi}
\mathcal{N}_\varphi&
=\varphi(\bar\theta_1-\bar\theta_2)
+\int  \widetilde{d\theta} \
\varphi(\theta-\bar{\theta}_1-i\pi/2)
\varphi(\theta-\bar{\theta}_2-i\pi/2)+\\
&+
\sum_{n=2}^\infty \int \widetilde{d\theta_1}\dots
 \widetilde{d\theta_n}\
\varphi(\theta_1-\bar{\theta}_1-i\pi/2)\varphi_{12}\dots
\varphi_{n-1,n}\ \varphi(\theta_n-\bar{\theta}_2-i\pi/2).
\end{split}
\end{equation}

On the other hand,
using the definitions \eqref{Njj} and \eqref{recursion2} and the
explicit form of $n_{1,2}$ given by \eqref{n1_szepen_kifejezve12}
we obtain the integral series
\begin{equation}
\begin{split}
\label{norm_factor_szepen_kifejezve2p}
\mathcal{N}_j&=m\cosh\bar\theta_j
+\int  \widetilde{d\theta} \
i\varphi(\theta-\bar{\theta}_j+i\pi/2)m\sinh\theta+\\
&+
\sum_{n=2}^\infty \int \widetilde{d\theta_1}\dots
 \widetilde{d\theta_n}\
i\varphi(\theta_1-\bar{\theta}_j+i\pi/2)\varphi_{12}\dots
\varphi_{n-1,n}\ m\sinh\theta_n\\
&\hspace{8cm}\text{with }j=1,2.
\end{split}
\end{equation}

With these notations the limiting matrix $A_\infty$ can be obtained as
\begin{equation}
\label{Ainfty2}
  A_{\infty}=\left(I- M
\right)^{-1}
\begin{pmatrix}
  \frac{1}{n_1} & 0 \\ 0 & \frac{1}{n_2}
\end{pmatrix}
=\frac{1}{  \mathcal{N}_1  \mathcal{N}_2+(  \mathcal{N}_1+  \mathcal{N}_2)\mathcal{N}_\varphi }
\begin{pmatrix}
  \mathcal{N}_2+\mathcal{N}_\varphi & \mathcal{N}_\varphi \\
\mathcal{N}_\varphi &   \mathcal{N}_1+\mathcal{N}_\varphi
\end{pmatrix}.
\end{equation}

Putting together \eqref{2p-kifejt}, \eqref{Alim1}, \eqref{norma1d} and
\eqref{furcsa-allitas}
we obtain our final formula
\begin{equation}
\label{final2p}
  \begin{split}
&  \bra{\bar\theta_1,\bar\theta_2}\ordo
  \ket{\bar\theta_1,\bar\theta_2}_L=
\vev{\ordo}_\eps+\frac{1}{ \mathcal{N}_1  \mathcal{N}_2+(
  \mathcal{N}_1+  \mathcal{N}_2)\mathcal{N}_\varphi }\Big\{ \\
&\hspace{2cm}+ (\mathcal{N}_2+\mathcal{N}_\varphi)\sum_{k=0}^\infty
\frac{1}{k!}\int \widetilde{d\theta_1}\dots \widetilde{d\theta_k} \
F^\ordo_{2(1+k),c}(\bt_1+i\pi/2,\theta_1,\dots,\theta_k)\\
&\hspace{2cm}+
(\mathcal{N}_1+\mathcal{N}_\varphi) \sum_{k=0}^\infty
\frac{1}{k!}\int \widetilde{d\theta_1}\dots \widetilde{d\theta_k} \
 F^\ordo_{2(1+k),c}(\bt_2+i\pi/2,\theta_1,\dots,\theta_k)\\
&\hspace{2cm}+
\sum_{k=0}^\infty
\frac{1}{k!}\int \widetilde{d\theta_1}\dots \widetilde{d\theta_k} \
F^\ordo_{2(2+k),c}(\bt_1+i\pi/2,\bt_2+i\pi/2,\theta_1,\dots,\theta_k)\Big\}.
  \end{split}
\end{equation}

\subsection{Interpretation of the result}

The quantities $\mathcal{N}_1$, $\mathcal{N}_2$ and
$\mathcal{N}_\varphi$ defined above appear to be dressed versions of
the elements of the Jacobian of the two-particle Bethe equations. In
the following we show that they are indeed the appropriate total
derivatives of the two-particle quantization conditions as defined by
the excited state TBA.

We define
\begin{equation*}
  \bar Q_j=-i\eps(\bar\theta_j+i\pi/2),\qquad\qquad j=1,2,
\end{equation*}
and
\begin{equation*}
\mathcal{K}_{jk}\equiv  \frac{\partial\bar Q_j}{\partial
  \bar\theta_k}.
\end{equation*}
Here it is understood that a total derivative has to be taken,
including the dependence of the pseudoenergy function on the roots $\bar\theta_{1,2}$.

The diagonal elements of the Jacobian are given by
\begin{equation*}
\mathcal{K}_{jj}=mL\cosh\bar\theta_j+\varphi(\bar\theta_1-\bar\theta_2)
-\int  \frac{d\theta'}{2\pi}i\varphi(\theta'-\bar\theta_j-i\pi/2)
\frac{1}{1+e^{\eps(\theta')}}
\left(\frac{\partial\eps}{\partial\theta }
+\frac{\partial\eps}{\partial\bar\theta_j }\right)
,\qquad j=1,2.
\end{equation*}
For the off-diagonal elements
\begin{equation*}
 \mathcal{K}_{jk}=   -\varphi(\bar\theta_j-\bar\theta_k)
-\int  \frac{d\theta'}{2\pi}i\varphi(\theta'-\bar\theta_j-i\pi/2)
\frac{1}{1+e^{\eps(\theta')}}\frac{\partial\eps}{\partial\bar\theta_k}
,\qquad j\ne k.
\end{equation*}
It follows from \eqref{2particleTBA} that the derivatives of the pseudoenergy satisfy
\begin{equation*}
\begin{split}
 (1-\hat{K})\frac{\partial\eps}{\partial\bar\theta_j }
 &=-i\varphi(\theta-\bar\theta_j-i\pi/2),\qquad j=1,2 \qquad\text{and}\\
\label{n12kakki22}
 (1-\hat{K})\frac{\partial\eps}{\partial\theta
 }&=mL\sinh\theta+i\varphi(\theta-\bar\theta_1-i\pi/2)+i\varphi(\theta-\bar\theta_2-i\pi/2).
\end{split}
\end{equation*}

Inverting the integral operator $\hat K$ and comparing with the
explicit formulas 
\eqref{norm_factor_szepen_kifejezve2p_phi}-\eqref{norm_factor_szepen_kifejezve2p}
we find
\begin{equation*}
\begin{split}
\mathcal{K}_{jk}\equiv  \frac{\partial \bar Q_j}{\partial \bar\theta_k}=
\begin{pmatrix}
  \mathcal{N}_1+\mathcal{N}_\varphi & -\mathcal{N}_\varphi \\
-\mathcal{N}_\varphi &  \mathcal{N}_2+\mathcal{N}_\varphi 
\end{pmatrix}.
\end{split}
\end{equation*}

Therefore our final result \eqref{final2p} is indeed
a ``dressed version'' of the asymptotic result \eqref{2p-asympt}, 
where instead of single form factors an appropriate integral
series appears, and the derivatives of the asymptotic Bethe equations
have been replaced by the total derivatives of the exact excited state quantization
conditions.

\section{Conjecture for the general multiparticle case}

\label{sec:final}

The results of the previous two sections suggest a simple
generalization to the higher particle case. We formulate our
conjecture below.

First it is useful to define a dressing function as follows. Given a certain
diagonal form factor $F_{2k}^\ordo(\bar\theta_1,\dots,\bar\theta_k)$ its
``dressed version'' is defined as
\begin{equation}
\label{Ddef}
\begin{split}
&\mathcal{D}^\ordo_\eps(\bar\theta_1,\dots,\bar\theta_k)\equiv
\sum_{n=0}^\infty \frac{1}{n!}
\int \frac{d\theta_1}{2\pi}\dots  \frac{d\theta_n}{2\pi}
\left(\prod_j \frac{1}{1+e^{\eps(\theta_j)}}\right)\\
&\hspace{5cm}
F^\ordo_{2(n+k),c}(\bar\theta_1+i\pi/2,\dots,\bar\theta_k+i\pi/2,\theta_1,\dots,\theta_n).
\end{split}
\end{equation}
The physical interpretation of this ``dressing'' is simply the
addition of virtual particles which wind around the finite volume
exactly one time. The imaginary shift of $i\pi/2$ corresponds to the
euclidean rotation of the particle world lines. A graphical
interpretation of the  dressing
operation is shown in figure \ref{figgg}. The pseudoenergy
$\eps(\theta)$ is not defined at this stage, it is simply a parameter
of the dressing operation.

\begin{figure}
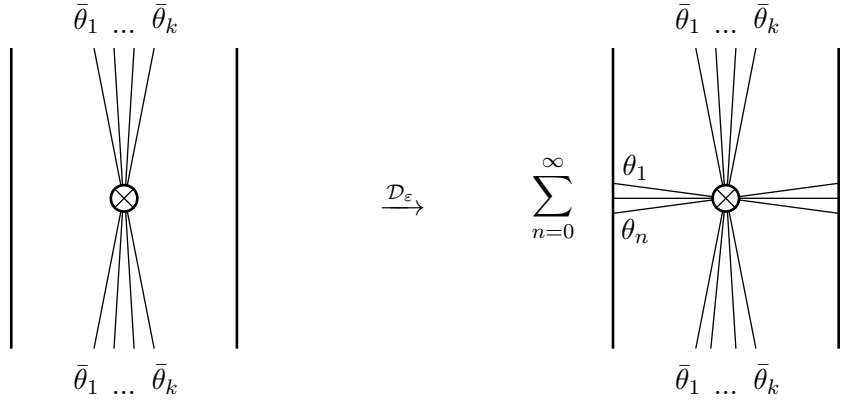

  \centering
\begin{pgfpicture}{0cm}{-1cm}{11cm}{4cm}
\pgfsetlinewidth{1pt}
    \pgfline{\pgfxy(0,0)}{\pgfxy(0,4)}
  \pgfline{\pgfxy(3,0)}{\pgfxy(3,4)}
\pgfcircle[stroke]{\pgfpoint{1.5cm}{2cm}}{5pt}
\pgfsetlinewidth{0.5pt}
\pgfline{\pgfxy(1.365,1.865)}{\pgfxy(1.635,2.135)}
\pgfline{\pgfxy(1.635,1.865)}{\pgfxy(1.365,2.135)}

\pgfline{\pgfxy(1.1,0)}{\pgfxy(1.46,1.84)}
\pgfline{\pgfxy(1.366,0)}{\pgfxy(1.48,1.84)}
\pgfline{\pgfxy(1.6333,0)}{\pgfxy(1.52,1.84)}
\pgfline{\pgfxy(1.9,0)}{\pgfxy(1.54,1.84)}

\pgfline{\pgfxy(1.1,4)}{\pgfxy(1.46,2.16)}
\pgfline{\pgfxy(1.366,4)}{\pgfxy(1.48,2.16)}
\pgfline{\pgfxy(1.6333,4)}{\pgfxy(1.52,2.16)}
\pgfline{\pgfxy(1.9,4)}{\pgfxy(1.54,2.16)}

\pgfputat{\pgfxy(1,-0.6)}{\pgfbox[center,bottom]{$\bar\theta_1$}}
\pgfputat{\pgfxy(1.5,-0.6)}{\pgfbox[center,bottom]{$...$}}
\pgfputat{\pgfxy(2.05,-0.6)}{\pgfbox[center,bottom]{$\bar\theta_k$}}

\pgfputat{\pgfxy(1,4.2)}{\pgfbox[center,bottom]{$\bar\theta_1$}}
\pgfputat{\pgfxy(1.5,4.2)}{\pgfbox[center,bottom]{$...$}}
\pgfputat{\pgfxy(2.05,4.2)}{\pgfbox[center,bottom]{$\bar\theta_k$}}

\pgfputat{\pgfxy(7.2,2)}{\pgfbox[center,center]{$\displaystyle\sum_{n=0}^\infty$}}
\pgfputat{\pgfxy(5.2,2)}{\pgfbox[center,center]{$\displaystyle\xrightarrow{\mathcal{D}_\eps}$}}

%
%

\pgfsetlinewidth{1pt}
\pgfline{\pgfxy(8,0)}{\pgfxy(8,4)}
\pgfline{\pgfxy(11,0)}{\pgfxy(11,4)}
\pgfcircle[stroke]{\pgfpoint{9.5cm}{2cm}}{5pt}
\pgfsetlinewidth{0.5pt}
\pgfline{\pgfxy(9.365,1.865)}{\pgfxy(9.635,2.135)}
\pgfline{\pgfxy(9.635,1.865)}{\pgfxy(9.365,2.135)}

\pgfline{\pgfxy(9.1,0)}{\pgfxy(9.46,1.84)}
\pgfline{\pgfxy(9.3,0)}{\pgfxy(9.48,1.84)}
\pgfline{\pgfxy(9.6333,0)}{\pgfxy(9.52,1.84)}
\pgfline{\pgfxy(9.9,0)}{\pgfxy(9.54,1.84)}

\pgfline{\pgfxy(9.1,4)}{\pgfxy(9.46,2.16)}
\pgfline{\pgfxy(9.366,4)}{\pgfxy(9.48,2.16)}
\pgfline{\pgfxy(9.6333,4)}{\pgfxy(9.52,2.16)}
\pgfline{\pgfxy(9.9,4)}{\pgfxy(9.54,2.16)}

\pgfputat{\pgfxy(9,-0.6)}{\pgfbox[center,bottom]{$\bar\theta_1$}}
\pgfputat{\pgfxy(9.5,-0.6)}{\pgfbox[center,bottom]{$...$}}
\pgfputat{\pgfxy(10.05,-0.6)}{\pgfbox[center,bottom]{$\bar\theta_k$}}

\pgfputat{\pgfxy(9,4.2)}{\pgfbox[center,bottom]{$\bar\theta_1$}}
\pgfputat{\pgfxy(9.5,4.2)}{\pgfbox[center,bottom]{$...$}}
\pgfputat{\pgfxy(10.05,4.2)}{\pgfbox[center,bottom]{$\bar\theta_k$}}

\pgfline{\pgfxy(8,2.2)}{\pgfxy(9.34,2.02)}
\pgfline{\pgfxy(8,2)}{\pgfxy(9.34,2)}
\pgfline{\pgfxy(8,1.8)}{\pgfxy(9.34,1.98)}

\pgfline{\pgfxy(11,2.2)}{\pgfxy(9.66,2.02)}
\pgfline{\pgfxy(11,2)}{\pgfxy(9.66,2)}
\pgfline{\pgfxy(11,1.8)}{\pgfxy(9.66,1.98)}

\pgfputat{\pgfxy(8.3,2.25)}{\pgfbox[center,bottom]{$\theta_1$}}
\pgfputat{\pgfxy(8.3,1.4)}{\pgfbox[center,bottom]{$\theta_n$}}

\end{pgfpicture}
\caption{Graphical interpretation of the ``dressing operation''
  $\mathcal{D}_\eps$. Time runs in the vertical direction. In the
  horizontal direction periodic boundary conditions are understood.}
\label{figgg}
\end{figure}

Define the derivative matrix of the excited state quantization
conditions and its determinant as
\begin{equation*}
\mathcal{K}_{jk}=\frac{\partial \bar Q_j}{\partial \bar\theta_k}
\qquad\qquad
  \bar\rho_K(\bar\theta_1,\dots,\bar\theta_K)=\det \mathcal{K}_{jk}.
\end{equation*}

For a given bipartite partition of the rapidities
\begin{equation*}
  \{\bar\theta_1,\dots,\bar\theta_K\} = \{\bar\theta_+\}\cup \{\bar\theta_-\}
\end{equation*}
\begin{equation*}
  \big|\{\bar\theta_+\}\big|=K-n \quad\text{and}\quad \big|\{\bar\theta_-\}\big|=n
\end{equation*}
define the restricted determinant
\begin{equation*}
  \bar\rho_{K-n}(\{\bar\theta_+\}|\{\bar\theta_-\})=\det \mathcal{K}_+,
\end{equation*}
where $\mathcal{K}_+$ is the sub-matrix of  $\mathcal{K}$ corresponding
to the particles in the set $\{\bar\theta_+\}$.

With these notations, our conjectured expression for the exact finite
volume 
mean values
 reads
\begin{equation}
    \label{full-multiparticle-result}
\begin{split}
&_L\bra{\bar\theta_1,\dots,\bar\theta_K}\ordo\ket{\bar\theta_1,\dots,\bar\theta_K}_L=\\
&\hspace{2cm}\frac{1}{\bar\rho_K(\bar\theta_1,\dots,\bar\theta_K)}
\sum_{\{\bar\theta_+\}\cup \{\bar\theta_-\}}
\mathcal{D}_\eps^\ordo\big(\{\bar\theta_-\}\big)
\bar\rho_{K-n}\big(\{\bar\theta_+\}|\{\bar\theta_-\}\big),
\end{split}
\end{equation}
where it is understood that the pseudoenergy $\eps$ entering the
dressing operation is the solution of the excited state TBA
corresponding to the state $\ket{\bar\theta_1,\dots,\bar\theta_K}_L$.
The above equation is a generalization of the one-particle and two-particle results
 \eqref{LM1p} and \eqref{final2p}.
Also, it can be regarded as the dressed version of the asymptotic formula
\eqref{fftcsa2-result2}. 

\section{Conclusions}

In this paper we studied excited state mean values in finite volume
integrable QFT. Although our ideas are general and should be
applicable to arbitrary models with diagonal scattering, 
for technical
reasons we restricted ourselves to theories with only one particle
species. Moreover, 
we only considered the
sinh-Gordon model because of the simplicity of its excited state TBA equations.

Our main results are \eqref{LM1p} and \eqref{final2p}
for the one-particle and two-particle mean values. In the previous
section we also presented the formula \eqref{full-multiparticle-result}
which is a conjectured generalization to arbitrary higher particle numbers.

Our calculations were based on two important conjectures. First of
all, the basis for the present work is the LeClair-Mussardo integral
series for the ground state mean values. Although this series is generally accepted
to be true, a rigorous derivation from first principles is not yet
known. The LM series was proven in \cite{LM-sajat} using
the finite volume expansion \eqref{fftcsa2-result2}. The present work
shows that a derivation in the other direction is also possible: 
the result \eqref{fftcsa2-result2} follows from the LM series as a
result of the analytic continuation and an appropriate summation
procedure. It is interesting that neither the LM series nor the
expansion \eqref{fftcsa2-result2} have been proven from first
principles yet.

Our second assumption was that there is an analytic continuation
procedure which connects a subset of the excited states
 to the ground state. This
lead to the intermediate result \eqref{LMexc}. Although such an
analytic continuation has been established for certain models, and it
is believed to exist in other models as well, the general existence
has not  been proven. Moreover it is not known if it exists in the
case of the sinh-Gordon model, which was our primary example. 

In order to justify our results a number of checks can be
performed. These include
\begin{itemize}
\item A careful investigation of the IR limit. In the sinh-Gordon
  model we calculated the first
  exponential corrections of order $e^{-mR}$ by an independent method
  and found agreement with the first terms of the IR expansion of
  \eqref{LM1p} and \eqref{final2p}. This calculation will be presented
  elsewhere.
\item Considering the trace of the
  energy-momentum tensor: $\ordo=\Theta=T_{\mu}^\mu$.  In this case
  there is an exact relation
  \begin{equation}
      \label{eq:diagonal_exact_ff}
  \bra{n}\Theta\ket{n}_L=\frac{E_n}{L}+\frac{dE_n}{dL}
  \end{equation}
valid for arbitrary finite volume states. The right hand side of this
equation can be evaluated using the known excited state TBA. On the
other hand, the l.h.s. can also be evaluated using our integral series
and the known form factors of $\Theta$ \cite{leclair_mussardo}. We performed this
comparison 
and found an exact
agreement. This provides an important consistency check of our
calculations. The calculation itself is a simple generalization of the
corresponding calculation presented in \cite{leclair_mussardo} for the
ground state mean value.
\item The investigation of the UV limit. Calculating the $mL\to 0$
  limit of all the integrals in the infinite series and summing up the
  leading contributions it should be possible to recover matrix
  elements calculated in the conformal limit.
We leave this problem to further research. 
\end{itemize}

The technical details of the calculation show that once 
it is established how to represent
individual particles in the excited state TBA then the rule to get the
mean values is to substitute the 
complex rapidities into formulas of the form
\eqref{full-multiparticle-result}. 
One-particle states are typically represented as complex
conjugate pairs of rapidities, possibly of a particle of different
type. In these cases the complex conjugate pairs have to be substituted
into \eqref{full-multiparticle-result};
this way a one-particle mean value will be given by our two-particle formula.
This is in complete accordance with the
findings of \cite{Pozsgay:2008bf} which used essentially the same idea
to obtain the leading exponential correction, the so-called
$\mu$-term. The study of such cases will be pursued in a future publication.

We find it remarkable that our final results for the one-particle and
two-particle mean values (eqs. \eqref{LM1p} and \eqref{final2p}) are
relatively simple and intuitive compared to the cumbersome way in
which they were derived from the starting point \eqref{LMexc}. Also,
it is quite remarkable that the normalization factor in the final
formula is simply the total derivative (or Jacobian) of the
exact quantization condition for the Bethe roots. This calls for an
alternative derivation of the present results.

\vspace{1cm}
{\bf Acknowledgements} 

\bigskip

We would like to thank G\'abor Tak\'acs for numerous discussions, for
suggesting the proof of Theorem 1 and for many
 useful comments on the manuscript. 
Also, the author is grateful to Patrick Dorey and Roberto Tateo for a
stimulating discussion at a conference in Bologna in 2011. 

Most of
this work was carried out while the author was employed by the
NWO/VENI grant 016.119.023 at the University of Amsterdam, the Netherlands.

\addcontentsline{toc}{section}{References}
\bibliography{../../pozsi-general}

\providecommand{\href}[2]{#2}\begingroup\raggedright\begin{thebibliography}{10}

\bibitem{Luscher-Fmu1}
M.~Luscher, ``{Volume Dependence of the Energy Spectrum in Massive Quantum
  Field Theories. 1. Stable Particle States},''
  \href{http://dx.doi.org/10.1007/BF01211589}{{\em Commun.Math.Phys.} {\bf 104}
  (1986)  177}.

\bibitem{Luscher-Fmu2}
M.~Luscher, ``{Volume Dependence of the Energy Spectrum in Massive Quantum
  Field Theories. 2. Scattering States},''
  \href{http://dx.doi.org/10.1007/BF01211097}{{\em Commun.Math.Phys.} {\bf 105}
  (1986)  153--188}.

\bibitem{Lellouch-Luscher}
L.~Lellouch and M.~Luscher, ``{Weak transition matrix elements from finite
  volume correlation functions},'' {\em Commun.Math.Phys.} {\bf 219} (2001)
  31--44,
\href{http://arxiv.org/abs/hep-lat/0003023}{{\tt arXiv:hep-lat/0003023
  [hep-lat]}}.

\bibitem{zam-tba}
{Al. B. Zamolodchikov}, ``{Thermodynamic Bethe Ansatz in relativistic models.
  Scaling three state Potts and Lee-Yang models},''
\href{http://dx.doi.org/10.1016/0550-3213(90)90333-9}{{\em Nucl. Phys.} {\bf
  B342} (1990)  695--720}.

\bibitem{klassen_melzer_tba1}
T.~R. Klassen and E.~Melzer, ``{Purely Elastic Scattering Theories and their
  Ultraviolet Limits},''
\href{http://dx.doi.org/10.1016/0550-3213(90)90643-R}{{\em Nucl. Phys.} {\bf
  B338} (1990)  485--528}.

\bibitem{DdV1}
C.~Destri and H.~J. de~Vega, ``New thermodynamic Bethe ansatz equations without
  strings,'' \href{http://dx.doi.org/10.1103/PhysRevLett.69.2313}{{\em Phys.
  Rev. Lett.} {\bf 69} (1992) no.~16, 2313--2317}.

\bibitem{DdV2}
C.~Destri and H.~J. de~Vega, ``Unified approach to Thermodynamic Bethe Ansatz
  and finite size corrections for lattice models and field theories,''
  \href{http://dx.doi.org/DOI: 10.1016/0550-3213(94)00547-R}{{\em Nucl. Phys.
  B} {\bf 438} (1995) no.~3, 413 -- 454}.

\bibitem{blz-ddv}
V.~V. Bazhanov, S.~L. Lukyanov, and A.~B. Zamolodchikov, ``{Integrable
  structure of conformal field theory. 2. Q operator and DDV equation},''
  \href{http://dx.doi.org/10.1007/s002200050240}{{\em Commun.Math.Phys.} {\bf
  190} (1997)  247--278},
\href{http://arxiv.org/abs/hep-th/9604044}{{\tt arXiv:hep-th/9604044
  [hep-th]}}.

\bibitem{DoreyTateoAnCont1}
P.~Dorey and R.~Tateo, ``{Excited states by analytic continuation of TBA
  equations},'' \href{http://dx.doi.org/10.1016/S0550-3213(96)00516-0}{{\em
  Nucl.Phys.} {\bf B482} (1996)  639--659},
\href{http://arxiv.org/abs/hep-th/9607167}{{\tt arXiv:hep-th/9607167
  [hep-th]}}.

\bibitem{DoreyTateoAnCont2}
P.~Dorey and R.~Tateo, ``{Excited states in some simple perturbed conformal
  field theories},''
  \href{http://dx.doi.org/10.1016/S0550-3213(97)00838-9}{{\em Nucl.Phys.} {\bf
  B515} (1998)  575--623},
\href{http://arxiv.org/abs/hep-th/9706140}{{\tt arXiv:hep-th/9706140
  [hep-th]}}.

\bibitem{ddv-exc}
C.~Destri and H.~de~Vega, ``{Nonlinear integral equation and excited states
  scaling functions in the sine-Gordon model},''
  \href{http://dx.doi.org/10.1016/S0550-3213(97)00468-9}{{\em Nucl.Phys.} {\bf
  B504} (1997)  621--664}, \href{http://arxiv.org/abs/hep-th/9701107}{{\tt
  arXiv:hep-th/9701107 [hep-th]}}.

\bibitem{takacs-sine-Gordon-NLIE-exc}
G.~Feverati, F.~Ravanini, and G.~Takacs, ``{Nonlinear integral equation and
  finite volume spectrum of Sine-Gordon theory},''
  \href{http://dx.doi.org/10.1016/S0550-3213(98)00747-0}{{\em Nucl.Phys.} {\bf
  B540} (1999)  543--586}, \href{http://arxiv.org/abs/hep-th/9805117}{{\tt
  arXiv:hep-th/9805117 [hep-th]}}.

\bibitem{takacs-sineG-NLIE-odd}
G.~{Feverati}, F.~{Ravanini}, and G.~{Tak{\'a}cs}, ``{Scaling functions in the
  odd charge sector of sine-Gordon/massive Thirring theory},''
  \href{http://dx.doi.org/10.1016/S0370-2693(98)01406-3}{{\em Physics Letters
  B} {\bf 444} (1998)  442--450},
  \href{http://arxiv.org/abs/arXiv:hep-th/9807160}{{\tt arXiv:hep-th/9807160}}.

\bibitem{blz-exc}
V.~V. Bazhanov, S.~L. Lukyanov, and A.~B. Zamolodchikov, ``{Integrable quantum
  field theories in finite volume: Excited state energies},''
  \href{http://dx.doi.org/10.1016/S0550-3213(97)00022-9}{{\em Nucl.Phys.} {\bf
  B489} (1997)  487--531}, \href{http://arxiv.org/abs/hep-th/9607099}{{\tt
  arXiv:hep-th/9607099 [hep-th]}}.

\bibitem{ads-cft-review}
N.~Beisert, C.~Ahn, L.~F. Alday, Z.~Bajnok, J.~M. Drummond, {\em et al.},
  ``{Review of AdS/CFT Integrability: An Overview},''
  \href{http://dx.doi.org/10.1007/s11005-011-0529-2}{{\em Lett.Math.Phys.} {\bf
  99} (2012)  3--32}, \href{http://arxiv.org/abs/1012.3982}{{\tt
  arXiv:1012.3982 [hep-th]}}.

\bibitem{essler-konik-review}
F.~H.~L. {Essler} and R.~M. {Konik}, ``{Applications of Massive Integrable
  Quantum Field Theories to Problems in Condensed Matter Physics},'' {\em
  eprint arXiv:cond-mat/0412421} (2004)  ,
  \href{http://arxiv.org/abs/arXiv:cond-mat/0412421}{{\tt
  arXiv:cond-mat/0412421}}.

\bibitem{Leclair:1996bf}
A.~Leclair, F.~Lesage, S.~Sachdev, and H.~Saleur, ``{Finite temperature
  correlations in the one-dimensional quantum Ising model},''
  \href{http://dx.doi.org/10.1016/S0550-3213(96)00456-7}{{\em Nucl. Phys.} {\bf
  B482} (1996)  579--612},
\href{http://arxiv.org/abs/cond-mat/9606104}{{\tt arXiv:cond-mat/9606104}}.

\bibitem{leclair_mussardo}
A.~Leclair and G.~Mussardo, ``{Finite temperature correlation functions in
  integrable QFT},''
  \href{http://dx.doi.org/10.1016/S0550-3213(99)00280-1}{{\em Nucl. Phys.} {\bf
  B552} (1999)  624--642},
\href{http://arxiv.org/abs/hep-th/9902075}{{\tt arXiv:hep-th/9902075}}.

\bibitem{Doyon:finiteTreview}
B.~Doyon, ``{Finite-temperature form-factors: A Review},'' {\em SIGMA} {\bf 3}
  (2007)  011,
\href{http://arxiv.org/abs/hep-th/0611066}{{\tt arXiv:hep-th/0611066
  [hep-th]}}.

\bibitem{fftcsa2}
B.~Pozsgay and G.~Takacs, ``{Form factors in finite volume II:disconnected
  terms and finite temperature correlators},''
  \href{http://dx.doi.org/10.1016/j.nuclphysb.2007.07.008}{{\em Nucl. Phys.}
  {\bf B788} (2008)  209--251},
\href{http://arxiv.org/abs/0706.3605}{{\tt arXiv:0706.3605 [hep-th]}}.

\bibitem{LM-sajat}
B.~{Pozsgay}, ``{Mean values of local operators in highly excited Bethe
  states},'' \href{http://dx.doi.org/10.1088/1742-5468/2011/01/P01011}{{\em J.
  Stat. Mech.} {\bf 2011} (2011)  P01011},
  \href{http://arxiv.org/abs/1009.4662}{{\tt arXiv:1009.4662 [hep-th]}}.

\bibitem{Essler:2009zz}
F.~H.~L. Essler and R.~M. Konik, ``{Finite-temperature dynamical correlations
  in massive integrable quantum field theories},''
  \href{http://dx.doi.org/10.1088/1742-5468/2009/09/P09018}{{\em J. Stat.
  Mech.} {\bf 0909} (2009)  P09018},
\href{http://arxiv.org/abs/0907.0779}{{\tt arXiv:0907.0779 [cond-mat.str-el]}}.

\bibitem{Saleur:1999hq}
H.~Saleur, ``{A comment on finite temperature correlations in integrable
  QFT},'' \href{http://dx.doi.org/10.1016/S0550-3213(99)00665-3}{{\em Nucl.
  Phys.} {\bf B567} (2000)  602--610},
\href{http://arxiv.org/abs/hep-th/9909019}{{\tt arXiv:hep-th/9909019}}.

\bibitem{takacs-lm}
I.~Szecsenyi, G.~Takacs, and G.~Watts, ``{One-point functions in finite
  volume/temperature: a case study},''
  \href{http://arxiv.org/abs/1304.3275}{{\tt arXiv:1304.3275 [hep-th]}}.

\bibitem{CastroAlvaredo:2002ud}
O.~A. Castro-Alvaredo and A.~Fring, ``{Finite temperature correlation functions
  from form factors},''
  \href{http://dx.doi.org/10.1016/S0550-3213(02)00409-1}{{\em Nucl. Phys.} {\bf
  B636} (2002)  611--631},
\href{http://arxiv.org/abs/hep-th/0203130}{{\tt arXiv:hep-th/0203130}}.

\bibitem{D22}
B.~Pozsgay and G.~Takacs, ``{Form factor expansion for thermal correlators},''
  \href{http://dx.doi.org/10.1088/1742-5468/2010/11/P11012}{{\em J. Stat.
  Mech.} {\bf 11} (2010)  12}, \href{http://arxiv.org/abs/1008.3810}{{\tt
  arXiv:1008.3810 [hep-th]}}.

\bibitem{Pozsgay:2009pv}
B.~Pozsgay, ``{Finite volume form factors and correlation functions at finite
  temperature},''
\href{http://arxiv.org/abs/0907.4306}{{\tt arXiv:0907.4306 [hep-th]}}.

\bibitem{takacs-szecsenyi-2p}
I.~Szecsenyi and G.~Takacs, ``{Spectral expansion for finite temperature
  two-point functions and clustering},''
  \href{http://dx.doi.org/10.1088/1742-5468/2012/12/P12002}{{\em J.Stat.Mech.}
  {\bf 1212} (2012)  P12002}, \href{http://arxiv.org/abs/1210.0331}{{\tt
  arXiv:1210.0331 [hep-th]}}.

\bibitem{fftcsa1}
B.~Pozsgay and G.~Takacs, ``{Form factors in finite volume I: form factor
  bootstrap and truncated conformal space},''
  \href{http://dx.doi.org/10.1016/j.nuclphysb.2007.06.027}{{\em Nucl. Phys.}
  {\bf B788} (2008)  167--208},
\href{http://arxiv.org/abs/0706.1445}{{\tt arXiv:0706.1445 [hep-th]}}.

\bibitem{Pozsgay:2008bf}
B.~Pozsgay, ``{Luscher's mu-term and finite volume bootstrap principle for
  scattering states and form factors},''
  \href{http://dx.doi.org/10.1016/j.nuclphysb.2008.04.021}{{\em Nucl. Phys.}
  {\bf B802} (2008)  435--457},
\href{http://arxiv.org/abs/0803.4445}{{\tt arXiv:0803.4445 [hep-th]}}.

\bibitem{takacs-dangerous-mu}
G.~Takacs, ``{Determining matrix elements and resonance widths from finite
  volume: The Dangerous $\mu$-terms},''
  \href{http://dx.doi.org/10.1007/JHEP11(2011)113}{{\em JHEP} {\bf 1111} (2011)
   113},
\href{http://arxiv.org/abs/1110.2181}{{\tt arXiv:1110.2181 [hep-th]}}.

\bibitem{smirnov_ff}
F.~A. Smirnov, ``{Form-factors in completely integrable models of quantum field
  theory},''
{\em Adv. Ser. Math. Phys.} {\bf 14} (1992)  1--208.

\bibitem{Babujian:2006km}
H.~M. Babujian, A.~Foerster, and M.~Karowski, ``{The form factor program: A
  review and new results - the nested SU(N) off-shell Bethe ansatz},''
  \href{http://dx.doi.org/10.3842/SIGMA.2006.082}{{\em SIGMA} {\bf 2} (2006)
  082},
\href{http://arxiv.org/abs/hep-th/0609130}{{\tt arXiv:hep-th/0609130}}.

\bibitem{Gaudin-LL-norms}
M.~Gaudin, ``La function d'onde de Bethe pour les mod\`eles exacts de la
  m\'ecanique statistique,'' {\em Commisariat \'a l'\'energie atomique, Paris}
  (1983)  .

\bibitem{korepinBook}
V.~Korepin, N.~Bogoliubov, and A.~Izergin, {\em Quantum inverse scattering
  method and correlation functions}.
\newblock Cambridge University Press, 1993.

\bibitem{klassen_melzer}
T.~R. Klassen and E.~Melzer, ``{On the relation between scattering amplitudes
  and finite size mass corrections in QFT},''
\href{http://dx.doi.org/10.1016/0550-3213(91)90566-G}{{\em Nucl. Phys.} {\bf
  B362} (1991)  329--388}.

\bibitem{takacs-palmai}
T.~Palmai and G.~Takacs, ``{Diagonal multi-soliton matrix elements in finite
  volume},'' \href{http://dx.doi.org/10.1103/PhysRevD.87.045010}{{\em
  Phys.Rev.} {\bf D87} (2013)  045010},
\href{http://arxiv.org/abs/1209.6034}{{\tt arXiv:1209.6034 [hep-th]}}.

\bibitem{Bender-Wu}
C.~M. Bender and T.~T. Wu, ``Anharmonic Oscillator,''
  \href{http://dx.doi.org/10.1103/PhysRev.184.1231}{{\em Phys. Rev.} {\bf 184}
  (1969)  1231--1260}.

\bibitem{Teschner:2007ng}
J.~Teschner, ``{On the spectrum of the Sinh-Gordon model in finite volume},''
  \href{http://dx.doi.org/10.1016/j.nuclphysb.2008.01.021}{{\em Nucl. Phys.}
  {\bf B799} (2008)  403--429},
\href{http://arxiv.org/abs/hep-th/0702214}{{\tt arXiv:hep-th/0702214}}.

\bibitem{klumper-pearce1}
A.~{Kl{\"u}mper} and P.~A. {Pearce}, ``{Analytic calculation of scaling
  dimensions: Tricritical hard squares and critical hard hexagons},''
  \href{http://dx.doi.org/10.1007/BF01057867}{{\em Journal of Statistical
  Physics} {\bf 64} (1991)  13--76}.

\bibitem{klumper-pearce2}
A.~{Kl{\"u}mper} and P.~A. {Pearce}, ``{Conformal weights of RSOS lattice
  models and their fusion hierarchies},''
  \href{http://dx.doi.org/10.1016/0378-4371(92)90149-K}{{\em Physica A
  Statistical Mechanics and its Applications} {\bf 183} (1992)  304--350}.

\bibitem{Delfino:1996nf}
G.~Delfino, P.~Simonetti, and J.~L. Cardy, ``{Asymptotic factorisation of form
  factors in two-dimensional quantum field theory},''
  \href{http://dx.doi.org/10.1016/0370-2693(96)01035-0}{{\em Phys. Lett.} {\bf
  B387} (1996)  327--333},
\href{http://arxiv.org/abs/hep-th/9607046}{{\tt arXiv:hep-th/9607046}}.

\end{thebibliography}\endgroup
\bibliographystyle{utphys}

\end{document}